\documentclass[12pt,a4paper]{article}

\usepackage{amsmath,amsthm,amssymb,amscd,a4wide}

\usepackage{dsfont}
\usepackage{mathrsfs}
\usepackage{setspace}
\usepackage{hyperref}

\newcommand{\ud}{\mathrm{d}}
\newcommand{\ui}{\mathrm{i}}
\newcommand{\ue}{\mathrm{e}}
\newcommand{\eins}{\mathds{1}}
\newcommand{\vl}{\boldsymbol{l}}
\newcommand{\vy}{\boldsymbol{y}}
\newcommand{\M}{\mathrm{M}}
\newcommand{\cD}{{\mathcal D}}
\newcommand{\cH}{{\mathcal H}}
\newcommand{\kz}{{\mathbb C}}

\DeclareMathOperator{\tr}{Tr}
\DeclareMathOperator{\sgn}{sgn}

\newtheorem{theorem}{Theorem}[section]
\newtheorem{lemma}[theorem]{Lemma}
\newtheorem{prop}[theorem]{Proposition}

\theoremstyle{definition}
\newtheorem{defn}[theorem]{Definition}

\numberwithin{equation}{section}

\begin{document}

\thispagestyle{empty}

\vspace*{1cm}

\begin{center}

{\LARGE\bf Many-particle quantum graphs: A review} \\

\vspace*{2cm}

{\large Jens Bolte}%
\footnote{E-mail address: {\tt jens.bolte@rhul.ac.uk}}

\vspace*{0.5cm}

Department of Mathematics\\
Royal Holloway, University of London\\
Egham, TW20 0EX\\
United Kingdom\\

\vspace*{1cm}

{\large and}

\vspace*{1cm}

{\large Joachim Kerner}%
\footnote{E-mail address: {\tt Joachim.Kerner@fernuni-hagen.de}} 

\vspace*{0.5cm}

Department of Mathematics and Computer Science\\
FernUniversit\"{a}t in Hagen\\
58084 Hagen\\
Germany\\

\end{center}

\vfill

\begin{abstract} 
In this paper we review recent work that has been done on quantum many-particle 
systems on metric graphs. Topics include the implementation of singular interactions, 
Bose-Einstein condensation, sovable models and spectral properties of some simple 
models in connection with superconductivity in wires.
\end{abstract}

\newpage

\section{Introduction}
Quantum graph models describe the motion of particles along the edges of a 
metric graph. They have become popular models in various areas of physics and
mathematics as they combine the simplicity of one-dimensional models with the
potential complexity of graphs. One-particle quantum graphs and their
applications are describe in detail in 
\cite{KotSmi99,GNUSMY06,Exnetal08,Berkolaiko:2013}. 

Many-particle quantum systems are of fundamental importance in condensed
matter as well as in statistical physics, see \cite{MR04,SchwablSM,cazalilla2011one}. 
In particular, phenomena like Bose-Einstein condensation, Anderson localisation and 
superconductivity have attracted much attention both in a phenomenological and 
a mathematical context. However, those phenomena are notoriously difficult to 
address, so that models that are promising to yield interesting results while still being
sufficiently accessible are in demand. This was a major reason to develop and 
study quantum many-particle models on graphs. Another reason lies in the growing 
importance of one-dimensional, nano-technological devices~\cite{QuantumWellsWiresDots,GG08}. 

Among early quantum graph examples are models of two particles with singular 
interactions on simple graphs \cite{MP95,CauCra07,Ha07,Ha08}, where some 
basic spectral properties were studied. Other approaches involve quantum field 
theory on graphs (see, e.g., \cite{BelMin06,Sch09}) where, due to the presence 
of vertices and the finite lengths of edges, translation invariance is broken. This 
leads to the presence of symmetry algebras that are of interest in their own 
right \cite{MinRag04}. In the context of quantum integrability, these symmetry 
algebras play a role in the construction of many-particle quantum models on 
graphs in which one can represent eigenfunctions in terms of a Bethe-ansatz 
\cite{CauCra07,BolGar17,BolGar17a}. In this context also extensive studies 
of non-linear Schr\"odinger equations (see, e.g., \cite{Noj14,Cau15}) are of 
particular interest.

The phenomenon of Anderson localisation, which is known to occur in a large
class of systems governed by random Schr\"odinger operators
\cite{SimonBookSchrödinger,stollmann2001caught}, has been investigated for 
interacting particles on graphs in \cite{Sab14}.

When particles are indistinguishable, the particle exchange symmetry has to be 
implemented. In three or more dimensions this leads to the well-known Fermi-Bose 
alternative. However, in lower dimensions more options may become available 
including, e.g., the possibility of anyons in two dimensions \cite{LeiMyr77}. 
In models of discrete quantum graphs the possible exchange symmetry representations 
were identified in \cite{HarKeaRob11,HarKeaRobSaw14}, and a whole range of exotic 
options were found.

In this paper we mainly review our own contributions to many-particle quantum
graphs. This includes the construction of two types of singular pair interactions.
The first one \cite{BKSingular} is closely related to vertices and can be seen
as a model of interactions between particles that is mediated by the presence
of an impurity (thought of as being located in a vertex); this type of interactions
is similar to the one introduced in \cite{MP95}. The second type of singular
interactions \cite{BKContact} are the more familiar $\delta$-pair interactions.
They are models of very short-range, or contact interactions. When implemented
for bosons, these interactions lead to a Lieb-Liniger gas \cite{LL63} on a graph,
and in the limit of hardcore interaction they lead to a Tonks-Giradeau gas 
\cite{G60}. For all of these models it has been shown that they can be rigorously 
implemented with self-adjoint Hamiltonians and it has been proven that their
spectra are discrete and the eigenvalues follow a Weyl law.

Due to a well-known theorem of Hohenberg~\cite{Hohenberg}, free Bose gases in one 
dimension are often said to not display Bose-Einstein condensation (BEC). This 
statement, however, is only true if in finite volume one imposes Dirichlet or 
other standard boundary conditions. It has been known though that non-standard 
boundary conditions may lead to cases where a finite number of eigenvalues are 
negative and remain so in the thermodynamic limit, such that in this limit a 
spectral gap below the continuum develops. Such a scenario then leads to BEC into 
the negative-energy ground state \cite{LanWil79,VerbeureBook}. A similar behaviour 
occurs for free bosons on graphs, and it is possible to fully characterise all 
boundary conditions where this is the case~\cite{BolteKernerBEC}. For a gas of 
bosons with pairwise repulsive hardcore interactions, a suitable Fermi-Bose mapping, 
however, shows that no condensation can be expected. Furthermore, it can be shown 
that arbitrarily small repulsive pair interactions prohibit a Bose gas on a graph
to condense into the free ground state~\cite{BolKer16}. 

In statistical mechanics solvable models play a significant role. In this context 
solvability refers to the fact that eigenfunctions of the Hamiltonan have a simple
representation in terms of a so-called Bethe ansatz \cite{Bet31,Gau14}. This form
of the eigenfunctions also leads to a characterisation of eigenvalues through 
finitely many secular equations. The Lieb-Liniger gas \cite{LL63} of $N$ bosons
with $\delta$-interactions on a circle is an example of a solvable model and its
version on an interval \cite{Gau71} can be seen as a first example where vertices
play a role. Vertices of degree two (and higher) present obstacles to the solvability
of models with $\delta$-interactions. A modification of the interactions that
preserves solvability when one vertex of degree two is present was found in
\cite{CauCra07}. The basic idea behind this construction can be extended to 
arbitrarily (finitely) many vertices of any (finite) degree \cite{BolGar17}, as
well as to any (finite) number of particles \cite{BolGar17a}.

In the final section we are concerned with a two-particle model on a simple 
non-compact quantum graph, namely the half-line $\mathbb{R}_+$, which can be thought 
of as a quantum wire. Besides singular interactions localised on the vertex at 
zero~\cite{KM16,EggerKerner} and contact interactions of the Lieb-Liniger type, 
we introduce a binding potential that leads to a pairing of the two particles 
\cite{KMBound, KernerElectronPairs,KernerInteractingPairs}. We also provide 
generalisations of this model by considering singular two-particle interactions  
whose locations are randomly distributed along the half-line~\cite{KernerRandomI}, 
and by taking into account surface defects in coupling the continuous half-line to 
a discrete graph~\cite{KernerSurfaceDefects}. In all of these cases we are mainly 
interested in describing spectral properties of the associated Hamiltonians. Using 
the acquired knowledge about the spectrum we are able to investigate Bose-Einstein 
condensation of pairs. These results can be seen as statements about superconductivity 
in quantum wires.

\section{Preliminaries}
\label{Sec1}
\subsection{One-particle quantum graphs}
A quantum graph is a metric graph $\Gamma$ with a differential operator 
that serves as Hamiltonian operator describing the motion of a particle 
along the edges of the graph, see \cite{GNUSMY06,Berkolaiko:2013}. 
A metric graph is a (finite) combinatorial graph with a metric structure 
that arises from assigning lengths to edges. Let $\mathcal{V}$ be the 
set of vertices and 
$\mathcal{E}=\mathcal{E}_\mathrm{int}\cup\mathcal{E}_\mathrm{ext}$ 
be the set of edges. Then every $e\in\mathcal{E}_\mathrm{int}$, an internal 
edge, is adjacent to two distinct vertices, and every 
$e\in\mathcal{E}_\mathrm{ext}$, an external edge, is adjacent to a single 
vertex. A metric structure is introduced by assigning finite lengths to 
internal edges; external edges are considered to be of infinite length. 
In this way every $e\in\mathcal{E}_\mathrm{int}$ is identified with an 
interval $[0,l_e]$ whereas every $e\in\mathcal{E}_\mathrm{ext}$ is 
identified with a copy of the real semi-axis $[0,\infty)$. Graphs without 
external edges, $\mathcal{E}=\mathcal{E}_\mathrm{int}$, are compact.

One now introduces functions on $\Gamma$,
\begin{equation}
\psi = (\psi_1,\dots,\psi_E)\ ,
\end{equation}
where $E=|\mathcal{E}|$ and $\psi_e :[0,l_e]\to\mathbb{C}$ for internal 
edges and $\psi_e :[0,\infty)\to\mathbb{C}$ for external edges. In this
way one defines the Hilbert space
\begin{equation}
L^2(\Gamma) := \bigoplus_{e\in\mathcal{E}_\mathrm{int}} L^2(0,l_e)
\bigoplus_{e\in\mathcal{E}_\mathrm{ext}} L^2(0,\infty)\ ,
\end{equation}
as well as the Sobolev spaces 
\begin{equation}
H^m(\Gamma) := \bigoplus_{e\in\mathcal{E}_\mathrm{int}} H^m(0,l_e)
\bigoplus_{e\in\mathcal{E}_\mathrm{ext}} H^m(0,\infty)\ .
\end{equation}
A Hamiltonian operator $H=-\Delta+V$ is a self-adjoint operator (in many cases with 
domain $\mathcal{D}\subset H^2(\Gamma)$) that acts on functions on an 
edge as
\begin{equation}
(H\psi)_e = -\psi_e'' + V_e\psi_e \ , 
\end{equation}
where $V=(V_1,\dots,V_E)$ is a potential function.
In many quantum graph models, however, one considers the case $V=0$.

In the following we shall restrict our attention to compact graphs, although
the examples in Section~\ref{Sec3} will be non-compact; the necessary
modifications are more or less obvious.
 
In order to characterise domains $\mathcal{D}$ of self-adjointness 
one has to impose boundary conditions at the vertices on functions 
in the domain. We denote boundary values of functions as
\begin{equation}
\psi_{bv} = \bigl(\psi_1(0),\dots\psi_E(0),\psi_1(l_1),\dots,
           \psi_E(l_E)\bigr) \ ,
\end{equation}
and of inward derivatives as 
\begin{equation}
\psi'_{bv} = \bigl(\psi'_1(0),\dots\psi'_E(0),-\psi'_1(l_1),\dots,
            -\psi'_E(l_E)\bigr) \ .
\end{equation}
Self-adjoint realisations of $H$ can be obtained as maximal symmetric 
extensions of the operator with minimal domain $C_0^\infty(\Gamma)$
(see, e.g., \cite{KosSch99}). Their domains can be uniquely parametrised 
in terms of an orthogonal projector $P$ and a self-adjoint map $L$, such 
that $P^\perp LP^\perp=L$, on the space $\mathbb{C}^{2E}$ of boundary 
values, as \cite{Kuc04}
\begin{equation}
\label{1Pdomain}
\mathcal{D}(P,L) = \bigl\{ \psi\in H^2(\Gamma):\ (P+L)\psi_{bv}+
  P^\perp\psi'_{bv}=0 \bigr\} \ .
\end{equation}
It is often useful to work with quadratic forms instead of self-adjoint 
operators, making use of the fact that a semi-bounded (from below) 
self-adjoint operator defines a unique, semi-bounded and closed quadratic
form, and vice versa~\cite{BEH08}. The form associated with a quantum graph 
Laplacian $-\Delta$ on the domain \eqref{1Pdomain} is \cite{Kuc04}
\begin{equation}
\mathcal{Q}[\psi] = \int_\Gamma |\nabla\psi|^2\ \ud x 
-\langle\psi_{bv},L\psi_{bv}\rangle_{\mathbb{C}^{2E}} \ ,
\end{equation}
with form domain
\begin{equation}
\label{1Pformdomain}
\mathcal{D}_{\mathcal{Q}} = \bigl\{ \psi\in H^1(\Gamma):
\ P\psi_{bv}=0 \bigr\} \ .
\end{equation}
The boundary conditions prescribed in \eqref{1Pdomain} and 
\eqref{1Pformdomain} do not necessarily respect the connectivity of 
the combinatorial graph. The latter will, however, be the case for 
local boundary conditions, where
\begin{equation}
P=\bigoplus_{v\in\mathcal{V}}P_v \quad\text{and}\quad
L=\bigoplus_{v\in\mathcal{V}}L_v \ ,
\end{equation}
and $P_v$, $L_v$ act on the subspace $\mathbb{C}^{d_v}$ of boundary
values on the edge ends adjacent to the vertex $v\in\mathcal{V}$. Here
$d_v$ is the degree of the vertex $v$.

A quantum graph Hamiltonian $H=-\Delta+V$ defined on a domain 
\eqref{1Pdomain} is self-adjoint, bounded from below, and has compact 
resolvent (note that the latter fails to hold for non-compact graphs). 
Hence its spectrum is real, bounded from below, discrete
and eigenvalues accumulate only at infinity. In the most relevant case
of $V=0$, one can characterise eigenvalues through a secular determinant.
One first defines a (vertex) scattering matrix 
\begin{equation}
S(k) : = -P-(L+\ui kP^\perp)^{-1}(L-\ui kP^\perp)\ ,
\end{equation}
where $k\in\mathbb{C}$ is such that $k^2$ is a spectral parameter for
$-\Delta$, and then a matrix
\begin{equation}
T(k) : = \begin{pmatrix} 0 & \ue^{\ui k\boldsymbol{l}} \\ 
\ue^{\ui k\boldsymbol{l}} & 0 \end{pmatrix}
\end{equation}
encoding the metric information about $\Gamma$. Here $\ue^{\ui k\boldsymbol{l}}$
is a diagonal $E\times E$ matrix with diagonal entries $\ue^{\ui kl_e}$, 
$e=1,\dots,E$. Defining $U(k):=S(k)T(k)$, one can show \cite{KS06} that $k^2$ is
a non-zero eigenvalues of $-\Delta$ of multiplicity $m(k)$, iff $k$
is a zero of
\begin{equation}
\label{1Psecular}
\det\bigl(\eins-U(k)\bigr)
\end{equation}
of order $m(k)$. An eigenvalue zero has to be treated separately, see
\cite{KS06,BolEnd09,BolEggSte15}. A similar, slightly more complicated condition 
can be obtained for operators of the form $H=-\Delta+V$, see \cite{BolEggRue15}.

The secular equation \eqref{1Psecular} can be used to derive a trace formula 
\cite{Rot83,KotSmi99,BolEnd09} that expresses spectral functions in terms of 
sums over periodic orbits on the graph.

\subsection{Many-particle kinematics}
Following the general construction of systems of several (distinguishable)
particles from given one-particle systems in quantum mechanics, the
Hilbert space of $N$ distinguishable particles on a metric graph $\Gamma$
is
\begin{equation}
\label{NHilbert}
\mathcal{H}_N = L^2(\Gamma) \otimes\cdots\otimes L^2(\Gamma) \ .
\end{equation}
Vectors in the tensor product are collections of functions 
$\psi_{e_1\dots e_N}\in L^2([0,l_{e_1}]\times\dots\times [0,l_{e_N}])$. These are
functions of $N$ variables describing the positions of the particles
on the $N$ edges $e_1,\dots,e_N$, which do not need to be all different.
In a slight abuse of notation we shall view these collections of functions
as functions on the domain
\begin{equation}
\label{Ndomain}
D^{(N)}_\Gamma := \bigcup_{e_1\dots e_N}
(0,l_{e_1})\times\dots\times (0,l_{e_N}) \ ,
\end{equation}
such that we shall also use the notation $\mathcal{H}_N = L^2(D^{(N)}_\Gamma)$.

$N$-particle observables are self-adjoint operators on $\mathcal{H}_N$.
An operator $O$ that respects the tensor product structure \eqref{NHilbert}
of the Hilbert space,
\begin{equation}
\label{freeop}
O = \sum_{i=1}^N \eins\otimes\dots\otimes\eins\otimes O_i 
\otimes\eins\dots\otimes\eins\ ,
\end{equation}
does not detect any correlations or interactions between particles.
A Hamiltonian describing particle interactions, therefore, cannot be
of this product form. In other words, particle interactions will be 
implemented by choosing a Hamiltonian that does not have the product 
structure. This can be achieved either in the form of, say, a potential
\begin{equation}
\label{2Ppotential}
V(x_1,\dots,x_N) = \sum_{i,j=1}^{N} V_p (x_i ,x_j)
\end{equation}
with explicit pair interactions. However, one can also implement
interactions by choosing an operator domain for an $N$-particle Laplacian
$-\Delta_N$, acting as
\begin{equation}
(-\Delta_N \psi)_{e_1\dots e_N} =
-\frac{\partial^2\psi_{e_1\dots e_N}}{\partial x_{e_1}^2}-\dots 
-\frac{\partial^2\psi_{e_1\dots e_N}}{\partial x_{e_N}^2}\ ,
\end{equation}
that does not respect the tensor product structure for the operator.

When the $N$ particles are indistinguishable, the exchange symmetry has 
to be implemented in the kinematic set up of the quantum system. 
If one adopts the Bose-Fermi alternative, the only relevant representations
of the symmetric group will be the totally symmetric one (for bosons) and 
the totally anti-symmetric one (for fermions). The quantum state
spaces then are the totally symmetric and the totally anti-symmetric 
subspaces $\mathcal{H}_{N,B}$ and $\mathcal{H}_{N,F}$, respectively, 
of \eqref{NHilbert}

\section{Singular pair interactions}
\label{Sec2}
A possible way of introducing interactions is to violate the tensor product structure 
\eqref{freeop} with boundary conditions, either at the boundaries of the domains
$[0,l_{e_1}]\times\dots\times [0,l_{e_N}]$, or at additional boundaries introduced for the
purpose of generating other types of interactions. Typically, such boundary conditions
will lead to singular interactions that can formally be expressed in terms of 
$\delta$-functions, see e.g.~\cite{BrascheExnerKuperinSeba92,Behrndt2013}.
\subsection{Vertex-induced singular interactions}
\label{Sec2.1}
Boundary conditions imposed at the boundaries of $[0,l_{e_1}]\times\dots\times [0,l_{e_N}]$
alone correspond to interactions that act when at least one particle sits in a vertex 
(corresponding to $x_{e_j}=0$ or $x_{e_j}=l_{e_j}$). Hence we say that such interactions 
are vertex induced. An example for a pair of particles on the same edge (of length $l$) 
would be the two-dimensional Laplacian plus a formal potential of the form
\begin{equation}
v(x_1,x_2)\bigl[ \delta(x_1)  + \delta(x_1-l) +  \delta(x_2)  + \delta(x_2-l) \bigr] \ .
\end{equation}
A version of such an interaction on a $Y$-shaped graph can be found in
\cite{MP95}.

Constructing $N$-particle Laplacians with boundary conditions is not as straight
forward as for one-particle Laplacians. The reason for this is that the minimal 
symmetric operator, which is an $N$-particle Laplacian with domain 
$C_0^\infty(D^{(N)}_\Gamma)$, does not have finite deficiency indices. For that reason 
it is more appropriate to construct self-adjoint realisations of the $N$-particle 
Laplacian via their associated sesqui-linear forms.

In the following we restrict our attention to $N=2$, noting that this case contains
all the essential steps in order to construct $N$-particle Hamiltonians with pair
interactions. As a first step we simplify the notation in that we define
\begin{equation}
\label{lscale}
 \psi_{e_1 e_2}(x_{e_1},y_{e_2}) = \psi_{e_1 e_2}(l_{e_1}x,l_{e_2}y)
\end{equation}
with $x,y\in (0,1)$. The $4E^2$ boundary values of functions 
$\psi\in H^1(D^{(2)}_\Gamma)$ and derivatives of functions
$\psi\in H^2(D^{(2)}_\Gamma)$ then are
\begin{equation}
\label{graphbv}
 \psi_{bv}(y) =
 \begin{pmatrix}\sqrt{l_{e_2}}\psi_{e_1 e_2}(0,l_{e_2}y) \\
 \sqrt{l_{e_2}}\psi_{e_1 e_2}(l_{e_1},l_{e_2}y) \\
 \sqrt{l_{e_1}}\psi_{e_1 e_2}(l_{e_1}y,0) \\
 \sqrt{l_{e_1}}\psi_{e_1 e_2}(l_{e_1}y,l_{e_2})
 \end{pmatrix} \qquad\text{and}\qquad
 \psi'_{bv}(y) =
 \begin{pmatrix}\sqrt{l_{e_2}}\psi_{e_1 e_2,x}(0,l_{e_2}y) \\
 -\sqrt{l_{e_2}}\psi_{e_1 e_2,x}(l_{e_1},l_{e_2}y)\\
 \sqrt{l_{e_1}}\psi_{e_1 e_2,y}(l_{e_1}y,0) \\
 -\sqrt{l_{e_1}}\psi_{e_1 e_2,y}(l_{e_1}y,l_{e_2})\end{pmatrix} \ .
\end{equation}
Here $y\in [0,1]$ and the indices $e_1 e_2$ run over all $E^2$ possible pairs
with $e_1 ,e_2 =1,\dots,E$. 

With the one-particle form domain \eqref{1Pformdomain} in mind we now
introduce bounded and measurable maps $P,L: [0,1] \to \M(4E^2,\kz)$ such 
that for a.e. $y \in [0,1]$,
\begin{enumerate}
\item $P(y)$ is an orthogonal projector,
\item $L(y)$ is a self-adjoint endomorphism on $\ker P(y)$.
\end{enumerate}
With these maps we can define the quadratic form,
\begin{equation}
\label{Qform2graph}
\begin{split}
 Q^{(2)}_{P,L}[\psi] 
   &:= \langle  \nabla\psi,\nabla\psi \rangle_{L^2 (D_\Gamma)} - 
         \langle \psi_{bv},L(\cdot)\psi_{bv} \rangle_{L^2(0,1)\otimes\kz^{4E^2}} \\
   &=\sum_{e_1 ,e_2 =1}^E \int_{0}^{l_{e_2}}\int_0^{l_{e_1}}\Bigl( \bigl|
      \psi_{e_1 e_2,x}(x,y) \bigr|^2 + \bigl| \psi_{e_1 e_2,y}
      (x,y) \bigr|^2 \Bigr)\ \ud x\,\ud y \\
   &\qquad\qquad -\int_0^1 \langle\psi_{bv}(y),L(y)\psi_{bv}(y)
       \rangle_{\kz^{4E^2}} \ \ud y \ ,
\end{split}
\end{equation}
and prove the following result \cite{BKSingular}.
\begin{theorem}
\label{2quadformgraph}
Given maps $P,L:[0,1]\to \M(4E^2,\kz)$ as above that are bounded and measurable,
the quadratic form \eqref{Qform2graph} with domain
\begin{equation}
\label{Defquadgraph}
 \cD_{Q^{(2)}} = \{ \psi \in H^1(D_\Gamma):\ P(y)\psi_{bv}(y)=0\ \text{for a.e.}\
 y\in [0,1] \}
\end{equation}
is closed and semi-bounded (from below).
\end{theorem}
The semi-bounded, self-adjoint operator associated with this form via the 
representation theorem for quadratic forms~\cite{BEH08} can be identified 
as a self-adjoint realisation of the Laplacian when its domain is contained in 
$H^2(D^{(2)}_\Gamma)$; in this case the form is said to be regular.

In order to identify regular forms we need to impose further restrictions on the
maps $P$ and $L$. The first one is that they are block-diagonal, in the form
\begin{equation}
\label{2block}
 M(y) = \begin{pmatrix} \tilde M(y) & 0 \\ 0 & \tilde M(y) \end{pmatrix} \ ,
\end{equation}
with respect to an arrangement of the components of \eqref{graphbv} where
the upper two components for all pairs $e_1,e_2$ are separated from the lower
two components. Then we obtain the following result  \cite{BKSingular}.
\begin{theorem}
\label{TheoremPL(y)2}
Let $L$ be Lipschitz continuous on $[0,1]$ and let $P$ be of the block-diagonal
form \eqref{2block}. Assume that the matrix entries of $\tilde P$ are in 
$C^3(0,1)$ and possess extensions of class $C^3$ to some interval
$(-\eta,1+\eta)$, $\eta>0$. Moreover, when 
$y\in [0,\varepsilon_{1}]\cup [l-\varepsilon_{2},l]$, with some 
$\varepsilon_{1},\varepsilon_{2} > 0$, suppose that $L(y)=0$ and that 
$\tilde P(y)$ is diagonal with diagonal entries that are either zero or one. 
Then the quadratic form $Q^{(2)}_{P,L}$ is regular.
The associated semi-bounded, self-adjoint operator is a Laplacian with domain
\begin{equation}
\label{bcgraph}
\cD_2 (P,L) := \{ \psi\in H^2(D^{(2)}_\Gamma):\ (P(y)+L(y))\psi_{bv}(y)+
P^\perp(y)\psi'_{bv}(y)=0\ \text{for a.e.}\ y\in [0,1] \}\ .
\end{equation}
\end{theorem}
Note the similarity of \eqref{bcgraph} with \eqref{1Pdomain}.

As one would expect from a quantum systems with a configuration space
of finite volume, the spectrum of a two-particle Laplacian with domain
\eqref{bcgraph} is discrete. Moreover, a Weyl law for the eigenvalue count
holds: Let $\lambda_n$, $n\in\mathbb{N}$, denote the eigenvalues of the 
operator, then
\begin{equation}
\label{Weylsing}
N(\lambda) := \# \{n\in\mathbb{N}:\ \lambda_n\leq\lambda\} \sim
\frac{\mathcal{L}^2}{4\pi}\ ,\quad \lambda\to\infty\ ,
\end{equation}
where $\mathcal{L}=l_{e_1}+\cdots+l_{e_E}$ is the total length of the
metric graph, see \cite{BKSingular}.

The constructions above can be carried over to bosonic or fermionic systems
in a straight forward manner; for details see \cite{BKSingular}.

\subsection{Contact interactions}
\label{sec2.2}
Realistic two-particle interactions are often of the form \eqref{2Ppotential}.
When the range of the interaction is small one can model the pair potential
with a Dirac-$\delta$, so that the formal $N$-particle Hamiltonian is
\begin{equation}
\label{contactint}
H_N = -\Delta_N +\alpha\sum_{i<j}\delta(x_i-x_j) \ .
\end{equation}
Here $\Delta_N$ denotes a non-interacting, self-adjoint realisation of the $N$-particle 
Laplacian and $\alpha\in\mathbb{R}$ is a constant determining the interaction
strength. In this way an interaction takes
place when (at least) two particles are at the same position and, therefore, one speaks 
of a contact interaction. In order to implement contact interactions in a self-adjoint
operator one has to impose boundary conditions along hyperplanes in the 
configuration space of $N$ particles that are characterised by equations $x_i=x_j$.
Contact interactions for bosons on a circle have, e.g., been studied in much detail in
the form of the Lieb-Liniger model \cite{LL63}, and for distinguishable particles on
infinite star graphs in \cite{Ha07,Ha08}.

A self-adjoint operator representing the formal expression \eqref{contactint} can
be defined as an extension of the $N$-particle Laplacian with domain 
$C_0^\infty(D_\Gamma^{(N)})$. This can be done much in the same way as 
above for the vertex-induced singular interactions after additional boundaries have 
been introduced to the domain \eqref{Ndomain}. As contact interactions require 
two particles to be on the same edge, components in \eqref{Ndomain} where 
$e_1,\dots,e_N$ are $N$ distinct edges do not contribute. Taking the example 
of $N=2$ as for the vertex-induced singular interactions above, one introduces 
the subdivision
\begin{equation}
\label{diagsquare}
D_{ee} := [0,l_e]\times [0,l_e] = D_{ee}^+\cup D_{ee}^-\ ,
\end{equation}
of diagonal domains, where 
\begin{equation}
\label{squarecut}
D_{ee}^+ := \{(x,y)\in D_{ee}:\ x\geq y\}\quad\text{and}\quad
D_{ee}^- := \{(x,y)\in D_{ee}:\ x\leq y\}\ .
\end{equation}
These subdivisions modify the total domain $D^{(2)}_\Gamma$, see 
\eqref{Ndomain}. The resulting domain with
the additional boundaries is denoted as $D^{\ast(2)}_\Gamma$.

Boundary values of components $\psi_{e_1e_2}$ of functions 
$\psi\in H^2(D^{\ast(2)}_\Gamma)$ and their derivatives are as in \eqref{graphbv} 
when $e_1\neq e_2$. For the remaining components, however, the additional 
boundaries lead to the boundary values
\begin{equation}
\label{bvdiag}
 \psi_{ee,bv}(y) := \begin{pmatrix}  \sqrt{l_e}\psi^{-}_{ee}(0,l_e y) \\ 
 \sqrt{l_e}\psi^{+}_{ee}(l_e,l_e y) \\ \sqrt{l_e}\psi^{+}_{ee}(l_e y,0) \\ 
 \sqrt{l_e}\psi^{-}_{ee}(l_e y,l_e) \\ \sqrt{l_e} \psi^{+}_{ee}(l_e y,l_e y) \\ 
 \sqrt{l_e}\psi^{-}_{ee}(l_e y,l_e y) \end{pmatrix}
 \qquad\text{and}\qquad
 \psi'_{ee,bv}(y) := \begin{pmatrix} \sqrt{l_{e}}\psi^{-}_{ee,x}(0,l_{e}y) \\ 
 -\sqrt{l_e}\psi^{+}_{ee,x}(l_e,l_e y) \\ \sqrt{l_e}\psi^{+}_{ee,y}(l_e y,0) \\ 
 -\sqrt{l_e}\psi^{-}_{ee,y}(l_e y,l_e) \\ \sqrt{2l_e}\psi^{+}_{ee,n}(l_e y,l_e y)\\ 
 \sqrt{2l_e}\psi^{-}_{ee,n}(l_e y,l_e y) \end{pmatrix}\ ,
\end{equation}
for $y \in [0,1]$. Here $\psi_{ee}^\pm:D_{ee}^\pm\to\mathbb{C}$ and
\begin{equation}
\label{nderdiag}
 \psi^\pm_{ee,n} := 
 \frac{\pm 1}{\sqrt{2}}\bigl(\psi^\pm_{ee,x}-\psi^\pm_{ee,y}\bigr) 
\end{equation}
is the normal derivative along the diagonal part of the boundary.

The space $\mathbb{C}^{n(E)}$, $n(E)=4E^2+2E$, of boundary values decomposes
into a $4E^2$-dimensional subspace $W_{vert}$ of vertex-induced boundary
values as in Section~\ref{Sec2.1}, and a $2E$-dimensional subspace
$W_{cont}$ of boundary values on diagonals associated with contact interactions. 
Introducing maps $P$ and $L$ on $[0,1]$ that take values in the orthogonal
projectors and self adjoint maps on $W_{vert}\oplus W_{cont}$, respectively, in 
the same way as in Section~\ref{Sec2.1}, their restrictions to $W_{vert}$ should
satisfy the same properties as above. The restrictions to the edge-$e$ subspace
of $W_{cont}$ should take the form 
\begin{equation}
P_{cont,e}=\frac{1}{2}\begin{pmatrix} 1 & -1 \\ -1 & 1 \end{pmatrix}
\quad\text{and}\quad
L_{cont,e}=-\frac{1}{2}\alpha(y)\eins_2\ ,
\end{equation}
where $\alpha:[0,1]\to\mathbb{R}$ is a possibly varying, Lipschitz-continuous 
interaction strength. With boundary conditions as described in \eqref{bcgraph} 
this choice implies continuity of functions across diagonals,
\begin{equation}
\label{fctcont}
\psi^{+}_{ee}(l_e y,l_e y) =  \psi^{-}_{ee}(l_e y,l_e y) \ ,
\end{equation}
and satisfies jump conditions for the normal derivatives,
\begin{equation}
\label{derjump}
\psi^+_{ee,n}(l_e y,l_e y) + \psi^-_{ee,n}(l_e y,l_e y) = 
\frac{1}{\sqrt{2}}\alpha(y)\psi^\pm_{ee}(l_e y,l_e y) \ .
\end{equation}
These conditions ensure a rigorous, self-adjoint realisation of the 
$\delta$-type contact interactions \eqref{contactint}. The operator
is a two-particle Laplacian with domain \eqref{bcgraph}, where
now
\begin{equation}
\label{bc4contact}
P = P_{vert}\oplus P_{cont} \quad\text{and}\quad
L = L_{vert}\oplus L_{cont} \ .
\end{equation}
Hardcore contact interactions correspond to Dirichlet conditions along 
all diagonal boundaries. These conditions follow from $\delta$-type 
interactions by taking the limit $\alpha\to\infty$.
For more detail see~\cite{BKContact}.

As in the case of vertex-induced singular interactions, the spectrum
of the two-particle Laplacian with domain \eqref{bcgraph} and 
\eqref{bc4contact} is discrete and the Weyl law \eqref{Weylsing} 
holds \cite{BKContact}.

\subsection{A Lieb-Linger model on graphs}
The contact interactions of Section~\ref{sec2.2} offer an opportunity to
extend the Lieb-Linger model of $N$ bosons with $\delta$-interactions on 
a circle to arbitrary metric graphs. Implementing bosonic symmetry first
requires to restrict the $N$-particle Hilbert space $\mathcal{H}_N$, see
\eqref{NHilbert}, to its totally symmetric subspace $\mathcal{H}_{N,B}$. The
projector $\Pi_B$ to that subspace acts on vector $\psi\in\mathcal{H}_N$ as
\begin{equation}
\label{Boseproject}
(\Pi_B\psi)_{e_1\dots e_N} = \frac{1}{N!}\sum_{\pi\in S_N}
\psi_{e_{\pi(1)}\dots e_{\pi(N)}}\ ,
\end{equation}
where $S_N$ denotes the symmetric group. In order to implement $\delta$-type
contact interactions one has to dissect the hyper-rectangles
$(0,l_{e_1})\times\dots\times (0,l_{e_N})$ with at least two coinciding edges,
$e_i=e_j$ with $i\neq j$, along the hyperplanes $x_{e_i}=x_{e_j}$. The resulting
configuration space is $D_\Gamma^{\ast(N)}$. On the hyperplanes we impose 
boundary conditions that are equivalent to \eqref{fctcont}--\eqref{derjump}. 

The remaining, vertex related boundary values can be simplified by making 
use of the particle exchange symmetry. For functions 
$\Psi\in H^1_B(D^{\ast(N)}_\Gamma)$ they are
\begin{equation}
\label{BVI}
\psi_{bv,vert}(\vy) = \begin{pmatrix} \sqrt{l_{e_{2}}\dots l_{e_{N}}} 
\psi_{e_{1}\dots e_{N}}(0,l_{e_{2}}y_{1},\dots,l_{e_{N}}y_{N-1}) \\ 
\sqrt{l_{e_{2}}\dots l_{e_{N}}} 
\psi_{e_{1}\dots e_{N}}(l_{e_{1}},l_{e_{2}}y_{1},\dots,l_{e_{N}}y_{N-1}) \end{pmatrix} \ ,
\end{equation}
and for derivatives,
\begin{equation}
\label{BVII}
\psi^{'}_{bv,vert}(\vy) = \begin{pmatrix} \sqrt{l_{e_{2}}\dots l_{e_{N}}}  
\psi_{e_{1}\dots e_{N},x^{1}_{e_{1}}}(0,l_{e_{2}}y_{1},\dots,l_{e_{N}}y_{N-1}) \\ 
-\sqrt{l_{e_{2}}\dots l_{e_{N}}} 
\psi_{e_{1}\dots e_{N},x^{1}_{e_{1}}}(l_{e_{1}},l_{e_{2}}y_{1},\dots,l_{e_{N}}y_{N-1})  
\end{pmatrix} \ ,
\end{equation}
where $\vy=(y_{1},\dots,y_{N-1})\in [0,1]^{N-1}$. 

Introducing maps $L_{vert},P_{vert}:[0,1]^{N-1}\to M(2E^N,\mathbb{C})$ in analogy to
\eqref{bc4contact}, we are now in a position to introduce the quadratic form
\begin{equation}
\label{QuadFormContactI}
\begin{split}
Q^{(N)}_{B}[\psi] 
  &= N \sum_{e_{1}\dots e_{N}}\int_{0}^{l_{e_{1}}}\dots\int_{0}^{l_{e_{N}}} 
    |\psi_{e_{1}\dots e_{N},x_{e_1}}(x_{e_{1}},\dots,x_{e_{N}})|^{2}\ \ud x_{e_N}\dots
    \ud x_{e_1} \\
  &\quad -N\int_{[0,1]^{N-1}}\langle\psi_{bv,vert},L_{vert}(\vy)\psi_{bv,vert} 
    \rangle_{\kz^{2E^{N}}} \ud\vy \\
  &\quad +\frac{N(N-1)}{2}\sum_{e_{2}...e_{N}}\int_{[0,1]^{N-1}} \alpha(y_1)\ 
    |\sqrt{l_{e_{2}}\dots l_{e_{N}}}\psi_{e_{2}e_{2}\dots e_{N}}
    (l_{e_2}y_1,\vl\vy)|^{2}\ \ud\vy \ ,
\end{split}
\end{equation}
where $\vl\vy=(l_{e_2}y_1,l_{e_3}y_2,\dots,l_{e_N}y_{N-1})$, with form domain
\begin{equation}
\begin{split}
\label{QNformdomain}
\cD_{Q^{(N)}_{B}} = \{\psi \in H^{1}_{B}(D^{\ast(N)}_\Gamma);\ P_{vert}(\vy)
\psi_{bv,vert}(\vy)=0\ \text{for a.e.}\ \vy\in [0,1]^{N-1}\}\ .
\end{split}
\end{equation}
The first two lines in \eqref{QuadFormContactI} define a bosonic $N$-particle
Laplacian with vertex-related boundary conditions, whereas the last line
yields pairwise, $\delta$-type contact interactions.

The hardcore limit, $\alpha\to\infty$ (see above), of the Lieb-Liniger gas is 
the so-called Tonks-Girardeau gas \cite{G60}.

\section{Bose-Einstein condensation}
\label{secBEC}
One of the most interesting questions arising for bosonic many-particle system is 
whether they show the phenomenon of Bose-Einstein condensation (BEC). This occurs
when below a critical temperature the particles condense into the same one-particle
state~\cite{PO56}. The original version of BEC \cite{EinsteinBEC} was found for free, i.e., 
non-interacting bosons in three dimensions that are confined to box of finite volume 
and whose wave functions satisfy standard conditions at the boundary of the box; it
occurs in the thermodynamic limit of increasing the particle number and the volume 
of the box while keeping the particle density fixed. One can readily show that this form 
of BEC does not occur in one dimension as long as standard boundary conditions are 
imposed. However, it has long been known that BEC 
for free bosons does occur in one dimension when the boundary conditions are changed 
in such a way that the free, one-particle Hamiltonian has a negative eigenvalue 
and in the thermodynamic limit a gap remains in the spectrum between
the ground state and the continuum above zero \cite{LanWil79,VerbeureBook}.

\subsection{Free bosons}
In a many-particle system of $N$ free bosons, the Hamiltonian is a symmetrised version
of an operator with the tensor product structure \eqref{freeop}. Its eigenvalues are 
of the form $k_{n_1}^2+\cdots +k_{n_N}^2$, where $k_n^2$ is an eigenvalue of the
one-particle Hamiltonian, which we assume to be a Laplacian with domain 
\eqref{1Pdomain}. The number of negative eigenvalues is controlled by the 
self-adjoint map $L$ in the characterisation of the domain \cite{BL10}, and this determines 
whether or not BEC is found in the thermodynamic limit. In this limit the volume
growth is achieved by stretching all edge lengths with the same factor, 
$l_e\mapsto\eta l_e$, $\eta > 0$. Hence, the thermodynamic limit can be performed by sending
the total length $\mathcal{L}=\sum_e l_e$ to infinity.

The first result required in order to prove BEC establishes a gap in the spectrum
\cite{BolteKernerBEC}.
\begin{prop}
Let $-\Delta$ be a one-particle Laplacian on a compact metric graph with domain
\eqref{1Pdomain}. Assume that $L$ has at least one positive eigenvalue and let
$L_{max}$ be the largest eigenvalue. Then the ground state eigenvalue 
$k^2_0(\mathcal{L})$ of the Laplacian at total length $\mathcal{L}$ converges 
to $-L_{max}$ in the thermodynamic limit $\mathcal{L}\to\infty$.
\end{prop}
In the grand canonical ensemble of statistical mechanics (see, e.g., \cite{SchwablSM}
for details), the density of particles $\rho_{n}(\beta,\mu_{\mathcal{L}})$ in an eigenstate with 
eigenvalue $k^2_n(\mathcal{L})$ is
\begin{equation}
\rho_{n}(\beta,\mu_{\mathcal{L}})=\frac{1}{\mathcal{L}}
\frac{1}{\ue^{\beta(k^2_n(\mathcal{L})-\mu_{\mathcal{L}})}-1}\ ,
\end{equation}
where $\beta=1/k_BT$ is the inverse temperature and $\mu_{\mathcal{L}}\leq k^2_0(\mathcal{L})$ 
is the so-called chemical potential which itself depends on $\mathcal{L}$. More explicitly, 
$\mu_{\mathcal{L}}$ is chosen such that 
\begin{equation}
\rho=\frac{1}{\mathcal{L}}\sum_{n=0}^{\infty}
\frac{1}{\ue^{\beta(k^2_n(\mathcal{L})-\mu_{\mathcal{L}})}-1}
\end{equation}
is the fixed density of particles on the graph for all values of $\mathcal{L}$. 
\begin{defn} 
\label{BECdef}
We say that an eigenstate with eigenvalue $k^2_n(\mathcal{L})$ is macroscopically 
occupied in the thermodynamic limit if
\begin{equation}
\limsup_{\mathcal{L}\to\infty} \rho_{n}(\beta,\mu_{\mathcal{L}}) > 0\ .
\end{equation}
If such an eigenstate exists we say that there is BEC into this eigenstate. 
\end{defn}
With these observations one is able to obtain a complete characterisation of free 
Bose gases on compact graphs in terms of BEC  \cite{BolteKernerBEC}.
\begin{theorem}
\label{freeBEC}
Let $\Gamma$ be a compact metric graph with one-particle Laplacian defined on the
domain \eqref{1Pdomain}. If $L$ is negative semi-definite, no BEC occurs at finite
temperature in the thermodynamic limit. 

If, however, $L$ has at least one positive eigenvalue, there exists a critical 
temperature $T_c >0$ such that BEC occurs below $T_c$ in the thermodynamic limit.
\end{theorem}
The one-particle ground state eigenfunction into which all particles condense 
below the critical temperature is peaked around the vertices and hence is not
homogeneous, as it would be in the classical case of particles in a box with
Dirichlet boundary conditions, see also \cite{LanWil79}.

\subsection{Interacting bosons}
For interacting bosons it is much harder to prove that BEC either holds or is 
absent~\cite{LiebSeiringerProof,LVZ03}. In the case of the Tonks-Girardeau 
gas \cite{G60} of particles with hardcore interactions on a graph, however, 
one can use a Fermi-Bose mapping in order to prove the absence 
of phase transitions which then indicates an absence of BEC. The Fermi-Bose 
mapping is a bijection between the set of bosonic many-particle 
Laplacians with hardcore interactions and the set of free fermionic Laplacians 
on the same compact, metric graph. The fermionic $N$-particle Hilbert space 
$\mathcal{H}_{N,F}$ is the totally antisymmetric subspace of $\mathcal{H}_N$, 
i.e., the image of the projector
\begin{equation}
\label{Fermiproject}
(\Pi_F\psi)_{e_1\dots e_N} = \frac{1}{N!}\sum_{\pi\in S_N}(-1)^{\sgn \pi}
\psi_{e_{\pi(1)}\dots e_{\pi(N)}}\ ,
\end{equation}
compare \eqref{Boseproject}. One notices that the antisymmetry implies that 
continuous fermionic functions vanish along diagonal hyperplanes 
$x_{e_i}=x_{e_j}$, where $e_i=e_j$ but $i\neq j$, as do functions in the 
domain of a bosonic Laplacian with hardcore interactions. Using appropriate 
permutations of edges one can construct a bijection between bosonic and 
fermionic functions in such a way that the latter are in the domain of a fermionic
quadratic form that is associated with a free fermionic Laplacians. As the forms
coincide, the Fermi-Bose mapping is isospectral. For details of the construction
we refer to \cite{BolteKernerBEC}.  In fermionic systems BEC is well known to be 
absent. In the present case one calculates the free-energy density
of free fermions (with Dirichlet boundary conditions in the vertices) in the 
thermodynamic limit,
\begin{equation}
f_{D,F}(\beta,\mu) = \limsup_{\mathcal{L}\to\infty}\frac{1}{\beta\mathcal{L}}
\tr\ue^{-\beta H_N}
= -\frac{1}{\beta}\int_0^\infty 
\log(1+\ue^{-\beta(k^2-\mu)})\,\ud k,
\end{equation}
see \cite{BolteKernerBEC}. This energy density is smooth and has no singularities
in $\beta$ which shows that there is no phase transition, consequently indicating 
an absence of BEC.

Other forms of (repulsive) interactions can be modelled by pair potentials of 
the type \eqref{2Ppotential}. On a metric graph this takes the form
\begin{equation}
\label{reppairint}
(V_{N,\mathcal{L}}\psi)_{e_1\dots e_N}(x_{e_1},\dots,x_{e_N}) =
\sum_{i<j} V_{p,\mathcal{L}}(x_{e_i}-x_{e_j})
\psi_{e_1\dots e_N}(x_{e_1},\dots,x_{e_N})\ ,
\end{equation}
and gives rise to the (bosonic) $N$-particle Hamiltonian
\begin{equation}
H_N = -\Delta_N + V_{N,\mathcal{L}}\ .
\end{equation}
The pair potentials are repulsive when the functions $V_{p,\mathcal{L}}$ are
non-negative, and for technical reasons we assume that for all $\mathcal{L}>0$
there exist $A_{\mathcal{L}}>0$ and $\epsilon_{\mathcal{L}}>0$ such that 
$V_{p,\mathcal{L}}(x)\geq\epsilon_{\mathcal{L}}$ for all 
$x\in[-A_{\mathcal{L}},A_{\mathcal{L}}]$. Moreover, the $L^1$-norm of
$V_{p,\mathcal{L}}$ is assumed to be independent of $\mathcal{L}$. These 
assumptions are consistent with choosing functions $V_{p,\mathcal{L}}$ that are
a $\delta$-series in the thermodynamic limit $\mathcal{L}\to\infty$. One can, 
e.g., take $V_{p,\mathcal{L}}(x)=\mathcal{L}v(\mathcal{L}x)$ with 
$v\in C_0^\infty(\mathbb{R})$, $v\geq 0$ and $\|v\|_1=\alpha$ so that
$\lim_{\mathcal{L}\to\infty}V_{p,\mathcal{L}}(x)=\alpha\delta(x)$. With this
choice the Lieb-Liniger model will be recovered in the thermodynamic limit.

A Gibbs state at inverse temperature $\beta>0$ is defined via
\begin{equation}
\omega_\beta (O) := \frac{\tr \left(O\,\ue^{-\beta H_N}\right)}{\tr\ue^{-\beta H_N}}\ ,
\end{equation}
where $O$ is a (bounded) observable, i.e., a (bounded and) self-adjoint operator.
If now $\psi_0$ is the ground state of the free bosonic system, i.e., composed of
the ground state eigenfunction $\phi_0$ of the one-particle Laplacian and $N(\phi_0)$ 
is the particle number operator in this ground state one can infer from
Theorem~\ref{freeBEC} whether or not the non-interacting system shows BEC.
Assuming this to be the case, one can ask what the effect of adding a repulsive 
interaction \eqref{reppairint} is. It can be shown \cite{BolKer16} that in the
theromodynamic limit the occupation of this ground state vanishes,
\begin{equation}
\limsup_{\mathcal{L}\to\infty}\frac{\omega_\beta(N(\phi_0))}{\mathcal{L}} =0\ .
\end{equation}
According to a direct analogue of Definition~\ref{BECdef}, this means that
there is no BEC into the free ground state. Hence, although BEC into the
ground state was present in the free bosonic system, even the smallest perturbation 
by repulsive pair interactions of the type \eqref{2Ppotential} make this 
condensation disappear.

Summarising, although free bosons on a compact metric graph may display
BEC, an addition of repulsive interactions is likely to destroy the condensate.
The BEC that can occur is caused by $\delta$-type, attractive, one-particle 
potentials in the vertices and the associated condensate is not homogeneous,
but concentrated around the vertices.

\section{Exactly solvable many-particle quantum graphs}
Much of the success of one-particle quantum graph models relies on the fact 
that eigenvalues possess a simple characterisation in terms of a secular
equation based on the finite-dimensional determinant \eqref{1Psecular}. On
the one hand this enables one to compute eigenvalues by searching for
zeros of a low-dimensional determinant, and on the other hand it leads to
a trace formula that is an identity \cite{Rot83,KotSmi99,BolEnd09} rather 
than a semiclassical approximation as in other, typical models of quantum
systems (see, e.g., \cite{Gut90}). 

The secular equation rests on the fact that the edge-$e$ component of an 
eigenfunction must be of the form
\begin{equation}
\psi_e (x_e) = a_e\,\ue^{\ui kx_e} + b_e\,\ue^{-\ui kx_e}\ ,
\end{equation}
with some coefficients $a_e,b_e\in\mathbb{C}$. It provides a sufficient
condition that the $2E$ coefficients must satisfy in order to yield an 
eigenfunction. Components of $N$-particle eigenfunctions with eigenvalue 
$\lambda$ are functions of $N$ variables,
$x_{e_1},\dots,x_{e_N}$, so that, in general, they are of the form
\begin{equation}
\label{Neigenfct}
\psi_{e_1\dots e_n}(x_{e_1},\dots,x_{e_N}) = \int_{\mathbb{R}^N}
a_{e_1\dots e_n} (k_1,\dots,k_N)\,\delta(k_1^2+\cdots +k_N^2 -\lambda)\,
\ue^{\ui(k_1x_{e_1}+\cdots +k_N x_{e_N})}\,\ud^N k \ .
\end{equation}
Hence, instead of the need to determine constants, in generic cases
with $N\geq 2$ a replacement for the secular equation needs to determine 
coefficient functions $a_{e_1\dots e_n} (\cdot)$. This would therefore
be a condition imposed on elements of an infinite dimensional space.

However, under certain circumstances such conditions may collapse
to a finite dimensional subspace. This, indeed, will be the case if certain
integrability conditions are satisfied which imply that eigenfunctions can 
be represented by a so-called Bethe-ansatz \cite{Gau14}. In essence,
a Bethe-ansatz is a finite sum of plane waves,
\begin{equation}
\label{Bethean}
 \psi_{\mathrm{Bethe}}(x_1,\dots,x_N) = \sum_{\alpha\in J}A_\alpha\,
 \ue^{\ui(k^\alpha_1 x_1+\cdots+k^\alpha_N x_N)}\ ,
\end{equation}
where $J$ is a finite index set, such that the vectors 
$(k^\alpha_1,\dots,k^\alpha_N)$ with
$(k^\alpha_1)^2+\cdots+(k^\alpha_N)^2=\lambda$ are drawn from a finite
subset of $\mathbb{R}^N$. Contrasting this with the general form 
\eqref{Neigenfct} of an $N$-particle eigenfunction suggests that a
Bethe-ansatz will only be possible under some strict conditions. These
integrability conditions (see, e.g., \cite{Gau14,CauCra07}) are also behind 
the Lieb-Liniger model, for which it has long been known that eigenfunctions 
can be characterised in terms of a finite number of coefficients \cite{LL63} 
and take a Bethe-ansatz form \eqref{Bethean}. The first example
of a quantum graph with a non-trivial vertex where a Bethe-ansatz was
shown to work is a particle on a line or ring with one vertex, where
non-Kirchhoff conditions are imposed \cite{CauCra07}. Since $N$ particles
on a graph have a configuration space that is composed of subsets of 
$\mathbb{R}^N$, a further class of example in this spirit where a 
Bethe-ansatz for the eigenfunctions is known to exist is given by the
Dirichlet- or Neumann Laplacian on a fundamental domain for the action
of a Weyl group \cite{Ber80}. Indeed, the mechanism behind these examples 
can be carried over to a class of quantum graph models, generalising
the approach of \cite{CauCra07}. This has been done in 
\cite{BolGar17,BolGar17a}, and in the following we will review those results.

The simplest example is that of two bosons on an interval $[0,l]$ with 
Dirichlet boundary conditions at the interval ends and a $\delta$-interaction 
\eqref{contactint} between the particles. This is a modification of the
Lieb-Liniger model first studied by Gaudin \cite{Gau71}. The two-particle
Hilbert space
\begin{equation}
\mathcal{H}_2 = L^2(0,l)\otimes L^2(0,l) \cong L(D)\ ,
\end{equation}
where $D$ is the square \eqref{diagsquare} that will be dissected as in
\eqref{squarecut}. Accordingly, $\psi^\pm\in L^2(D^\pm)$, for which the
Bethe-ansatz
\begin{equation}
\label{BetheGaudin}
 \psi^\pm(x_1,x_2) = \sum_{P\in\mathcal{W}_2}A_P^\pm\,
 \ue^{\ui(k_{P(1)} x_1 + k_{P(2)}x_2)}\ ,
\end{equation}
can be shown to lead to eigenfunctions. Here $\mathcal{W}_2$ is a 
Weyl group, which is a finite group with eight elements. The fact that
the ansatz \eqref{BetheGaudin} is consistent comes from the conditions 
an eigenfunction has to satisfy:
\begin{itemize}
\item[(i)] $-\Delta\psi=\lambda\psi$;
\item[(ii)] $\psi(x_1,x_2)=\psi(x_2,x_1)$;
\item[(iii)] $\left(\frac{\partial}{\partial x_1}-\frac{\partial}{\partial x_2}
\right)\psi(x,x)=\alpha\psi(x,x)$;
\item[(iv)] $\psi(0,x)=\psi(l,x)=0$.
\end{itemize}
These conditions are compatible with the plane-wave form 
$A\,\ue^{\ui(k_1 x_1 + k_2 x_2)}$ and only require substitutions of the wave vectors
$(k_1,k_2)$ with either $(k_2,k_1)$, $(-k_1,k_2)$, or combinations thereof. 
These operations, seen as an action of a group on $\mathbb{R}^2$, generate 
the action of the Weyl group $\mathcal{W}_2=\mathbb{Z}/2\mathbb{Z}\rtimes S_2$.
An interesting interpretation of this in terms of reflected rays can be found 
in \cite{McG64}. The conditions (i)--(iv) also yield a restriction on the 
allowed wave vectors,
\begin{equation}
\ue^{-2\ui k_n l}=\frac{k_n+k_m-\ui\alpha}{k_n+k_m+\ui\alpha}\,
\frac{k_n-k_m-\ui\alpha}{k_n-k_m+\ui\alpha}\ ,
\end{equation}
for all $n\neq m\in\{1,2\}$. Solutions $(k_1,k_2)\neq(0,0)$ with 
$0\leq k_1\leq k_2$ then give eigenvalues $\lambda=k_1^2+k_2^2$.

The above model, for $N$ bosons, was first studied by Gaudin \cite{Gau71,Gau14}.
The original Lieb-Liniger model \cite{LL63}, however, was formulated for particles
on a circle. Instead of the Dirichlet conditions (iv) one then has to require
periodic boundary conditions, which renders the reflection  
$(k_1,k_2)\mapsto (-k_1,k_2)$ expendable. The Bethe ansatz for the Lieb-Liniger
model, therefore, only requires a summation over the symmetric group $S_2$,
rather than over the Weyl group $\mathcal{W}_2=\mathbb{Z}/2\mathbb{Z}\rtimes S_2$
as in \eqref{BetheGaudin}. Hence, one concludes that the boundaries of the 
interval are responsible the additional reflections necessary in the Bethe 
ansatz. In a graph language, the interval ends are vertices of degree one.
Adding a vertex of degree two in the context of a Bethe ansatz was first done
in \cite{CauCra07}, where is was found that this is incompatible with 
$\delta$-pair interactions. Instead, the interactions were modified to include
another contribution that formally looks like $\delta(x_1+x_2)$. This means that
the particles do not only interact when they touch, but also when they are the 
same distance away from the vertex on either of the edges connected by the vertex.
If then this interaction is provided with a variable strength that is supported 
in a neighbourhood of the vertex, this will still be a localised interaction.

An extension of this method to arbitrary metric graphs with generalisations of 
the interactions introduced in \cite{CauCra07} has been done in \cite{BolGar17},
and an extension to $N$ particles can be found in \cite{BolGar17a}. The first step 
is to define the singular pair interactions, and this is most clearly done on
a star graph of $d$ half-lines. Then a given metric graph is first converted 
into its star representation, consisting of $|\mathcal{V}|$ star graphs, i.e., 
one for each $v\in\mathcal{V}$ of degree $d_v$. A Bethe ansatz is
made for each star graph, and then the boundary conditions that represent the
pair interactions on each star, as well as the matching conditions that allow
to recover the original graph from its star representation imply conditions
that characterise eigenvalues of the Laplacian with the singular pair
interactions.  

If now $\Gamma_v$ is the star graph with $d_v$ half-lines that is associated 
with the vertex $v\in\mathcal{V}$, the Hilbert space is $\oplus_{ee'}L^2(D_{ee'})$; 
here $e$ and $e'$ are edge labels and $D_{ee'}=\mathbb{R}_+^2$ is the two-particle    
configuration space when one particle is on edge $e$ and the other one on $e'$.
These configuration spaces are dissected into $D_{ee'}^+$ and $D_{ee'}^-$, which are
defined in analogy to \eqref{squarecut}, and the restrictions of functions 
$\psi_{ee'}$ to $D_{ee'}^\pm$ are denoted as $\psi_{ee'}^\pm$. One then requires that
\begin{equation}
\label{deltatildecond}
\begin{split}
\psi^+_{ee'}(x,x) &=\psi^-_{e'e}(x,x)\ ,\\
\left(\frac{\partial}{\partial x_1}-\frac{\partial}{\partial x_2}
-2\alpha\right)\psi^+_{ee'}(x,x)
                &=\left(\frac{\partial}{\partial x_1}-
                                 \frac{\partial}{\partial x_2}\right)
                                 \psi^-_{e'e}(x,x)\ .
\end{split}
\end{equation}
These conditions are similar to those generating $\delta$-interactions. However,
they apply to all pairs of edges, not only the diagonal ones. Hence there is a
singular interaction, also across edges, whenever two particles are the same distance 
away from the 
vertex. A Bethe ansatz \eqref{BetheGaudin} is then introduced for the functions
$\psi_{ee'}^\pm$, with the yet to be determined coefficients $A^\pm_{P,ee'}$. In a 
next step one has to cut the edges of the stars to the finite lengths that are 
required and then glue the stars to finally yield the original compact graph. In
this glueing process it has to be ensured that the interactions only take place
when two edges are connected in the same vertex, and not arbitrarily across the
graph. In addition to \eqref{deltatildecond}, this yields conditions to be imposed 
on the coefficients $A^\pm_{P,ee'}$. These conditions can be formulated in terms
of secular equations involving determinants
\begin{equation}
\label{sectildedelta}
Z(k_1,k_2) = \det\bigl(\eins-U(k_1,k_2)\bigr),
\end{equation}
where
\begin{equation}
U(k_1,k_2) = E(k_2)Y(k_2-k_1)(\eins_2\otimes S(k_2)\otimes\eins_{2E})Y(k_1+k_2),
\end{equation}
and
\begin{equation}
\begin{split}
Y(k)&=\frac{1}{k+i\alpha}\begin{pmatrix}-i\alpha&k\\k&-i\alpha\end{pmatrix}\otimes
            \boldsymbol{\alpha}+\begin{pmatrix}0&1\\1&0\end{pmatrix}
            \otimes(\eins_{E^2}-\boldsymbol{\alpha})\mathbb{T}_{E^2} \\
E(k)&=\eins_{4E}\otimes\begin{pmatrix}0&1\\1&0\end{pmatrix}\otimes
            \ue^{\ui k\boldsymbol{l}};
\end{split}
\end{equation}
Here $\mathbb{T}_{E^2}$ is a permutation matrix, and $\boldsymbol{\alpha}$ is
a diagonal matrix with the interaction strengths (which could, in principal, be 
different for each pair of edges) $\alpha_{ee'}$ on the diagonal. More details
can be found in \cite{BolGar17}. The final result is the following statement.
\begin{theorem}
\label{deltatildesecularThm}
Let $-\Delta_2$ be a two-particle Laplacian on a compact metric graph with
pair interactions as decribed above. Then the zeros $(k_1,k_2)$, where 
$0\leq k_1\leq k_2$, of $Z(k_i,k_j)$ for $i\neq j\in\{1,2\}$ of order $m$
correspond to eigenvalues $k_1^1+k_2^2$ of $-\Delta_2$.
\end{theorem}
We note that, with \eqref{sectildedelta} in mind, the secular equations are
reminiscent of the one-particle case \eqref{1Psecular}. Some special cases and 
numerical results in some example can be found in \cite{BolGar17}. The 
generalisation to $N$ particles follows the same lines and is contained in
\cite{BolGar17a}.

\section{Many-particle models on a simple non-compact graph}
\label{Sec3}
In this section we are concerned with interacting two-particle systems on a simple non-compact 
quantum graph, namely the positive half-line $\mathbb{R}_+=[0,\infty)$. More specifically, the Hamiltonian 
has several contributions: a hard-wall binding potential and two singular contributions, one of which
is of the vertex-induced type defined in Section~\ref{Sec2.1} and the other one representing the contact 
interactions introduced in Section~\ref{sec2.2}. The Hamiltonian is formally given by
\begin{equation}
\label{FormalHamiltonian}
	H=-\frac{\partial^2}{\partial x^2}-\frac{\partial^2}{\partial y^2}+v_{b}(|x-y|)+v(x,y)\left[\delta(x)+\delta(y)\right]\ + \alpha(y) \delta(x-y)\ ,
\end{equation}
where $v_{b}:\mathbb{R}_+ \rightarrow \overline{\mathbb{R}_+}$ is a (hard-wall) binding 
potential that is (formally) defined via
\begin{equation}
	v_{b}(x):=\begin{cases} 0 \quad \text{for} \quad x\leq d \ , \\
		\infty \quad \text{otherwise}\ ,
	\end{cases}
\end{equation}
where $d > 0$ characterises the size of the pair. We realise this formal potential by requiring
Dirichlet boundary conditions at $|x-y|=d$. Furthermore, $v:\mathbb{R}^2_+ \rightarrow \mathbb{R}$ 
is supposed to be a real-valued, symmetric and bounded potential, $v \in L^{\infty}(\mathbb{R}^2_+)$.
Note that setting $d=\infty$ corresponds to the case where no binding potential 
in~\eqref{FormalHamiltonian} is added. Also note that we always assume 
$\alpha(\cdot) \in L^{\infty}(\mathbb{R}_+)$.

It is important to note that interactions of the form~\eqref{FormalHamiltonian} generically break translation invariance, even with potentials $v(x,y)=v(|x-y|)$, and consequently lead to non-separable many-body problems. Although only rarely discussed in the literature, they have important applications in various areas of physics~\cite{glasser1993solvable,glasser2005solvable}. 

Other important situations in which singular interactions of the above form are expected can be found
in solid-state physics. 
For example, similar to the Cooper pairing mechanism of superconductivity~\cite{CooperBoundElectron}, 
two electrons in a metal can effectively interact with each other through the interaction of each individual 
particle with the lattice via electron-phonon-electron interactions. Hence, if a metal exhibits spatially
localised defects, there will be effective, spatially localised two-particle interactions. 

Furthermore, the idea to consider a binding potential in~\eqref{FormalHamiltonian}, which effectively 
leads to a `molecule', or a pair of particles, also originated from Cooper's 
work \cite{CooperBoundElectron}; another example can be found in \cite{QUnruh}, where 
the scattering of a bound pair of particles at mirrors is investigated. As a matter of fact, it was Cooper 
who realised that electrons in a metal will form pairs (Cooper pairs), if the metal is in the superconducting
phase, i.e., is cooled below some critical temperature~\cite{BCSI,MR04}. Hence, the 
Hamiltonian~\eqref{FormalHamiltonian}, or versions thereof, provide toy models to investigate bound 
pairs of particles in a quantum wire with defects~\cite{KMBound,KernerRandomI}. Most importantly, 
in this model one can derive rigorous results related to the superconducting behaviour of quantum 
wires~\cite{KernerElectronPairs,KernerSurfaceDefects,KernerInteractingPairs}. 

\subsection{The model without hard-wall binding potential}
In this subsection we present results regarding the Hamiltonian~\eqref{FormalHamiltonian} without 
binding potential, i.e., $v_{b}\equiv 0$. For more detail, we refer 
to~\cite{BKSingular,BKContact,KM16,EggerKerner} from which most of the results are taken. 

In a first step one has to give a rigorous meaning to the Hamiltonian~\eqref{FormalHamiltonian}, 
which is only formally defined due the $\delta$-distributions. This requires a suitable variant of
Theorem~\ref{2quadformgraph} and Section~\ref{sec2.2}. (Note that here the two-particle configuration
space without binding potential is $\mathbb{R}^2_+$). In order to do this one constructs a quadratic 
form on $L^2(\mathbb{R}^2_+)$, 
\begin{equation}\label{FormNoBinding}
	q^{\infty}_{\alpha,\sigma}[\varphi]:=\int_{\mathbb{R}^2_+}|\nabla \varphi|^2 \ \ud x - \int_{0}^{\infty}
	\sigma(y)|\gamma(\varphi)|^2 \ud y + \int_{0}^{\infty}\alpha(y)|\varphi(y,y)|^2 \ \ud y\ ,
\end{equation}
where $\sigma(y):=-v(0,y)$ is a real-valued boundary potential and $\gamma(\varphi):=\left(\varphi(y,0),
\varphi(0,y)\right)^T$ are the boundary values of $\varphi \in H^1(\mathbb{R}^2_+)$, which are 
well-defined in $L^2(\partial \mathbb{R}^2_+)$ due to the trace theorem for Sobolev 
functions~\cite{Dob05}. In the same way one defines $\varphi|_{x=y}$ as the trace of 
$\varphi \in H^1(\mathbb{R}^2_+)$ along the diagonal $x=y$.
\begin{theorem} For any given $\sigma,\alpha \in L^{\infty}(\mathbb{R}_+)$ the form $\left(q^{\infty}_{\alpha,\sigma},H^1(\mathbb{R}^2_+)  \right)$ is bounded from below and closed. 
\end{theorem}
Hence, according to the representation theorem for quadratic forms~\cite{BEH08} there exists a unique 
self-adjoint operator associated with the form $q^{\infty}_{\alpha,\sigma}$. We denote this operator by 
$-\Delta^{d=\infty}_{\sigma,\alpha}$. Since the only volume term in~\eqref{FormNoBinding} is associated 
with the $\nabla$-operator, this operator acts as the standard two-dimensional Laplacian $-\Delta$ on 
functions $\varphi \in \cD(-\Delta^{d=\infty}_{\sigma,\alpha}) \subset H^1(\mathbb{R}^2_+)$. The boundary
integrals in~\eqref{FormNoBinding}, on the other hand, reflect boundary conditions. More explicitly, one 
has coordinate-dependent Robin conditions along $\partial \mathbb{R}^2_+$, and coordinate-independent
jump conditions along the diagonal $x=y$, see~\cite{BKContact,KM16} for more detail.

In a next step we characterise the spectrum of the self-adjoint operator 
$-\Delta^{d=\infty}_{\sigma,\alpha}$. 
\begin{theorem}
\label{TheoremEssentialSpectrumI} 
For any given $\sigma,\alpha \in L^{\infty}(\mathbb{R}_+)$ one has 
$[0,\infty) \subset \sigma_{ess}(-\Delta^{d=\infty}_{\sigma,\alpha})$. Furthermore, if 
$\sigma(y),\alpha(y) \rightarrow 0$ as $y \rightarrow \infty$ one has 
$\sigma_{ess}(-\Delta^{d=\infty}_{\sigma,\alpha=0})=[0,\infty)$. 
\end{theorem}
\begin{proof} 
See the proof of [Theorem~3.1,\cite{KM16}] for the case where $\alpha=0$. An inspection of this 
proof then allows one to conclude the statement above. 
\end{proof}
The discrete part of the spectrum, i.e., isolated eigenvalues with finite multiplicity, is characterised in 
the following statement. 
\begin{theorem}
Assume that $\sigma,\alpha \in L^1(\mathbb{R}_+)$ and that 
$\inf \sigma_{ess}(-\Delta^{d=\infty}_{\sigma,\alpha})=0$. Then, if
\begin{equation}
		\int_{\mathbb{R}_+} \left[2\sigma(y)-\alpha(y)\right] \ \ud y > 0\ ,
\end{equation}
negative eigenvalues will exist. 
\end{theorem}
\begin{proof} As in the proof of [Theorem~3.3,\cite{KM16}] one picks the test function 
$\varphi_{\varepsilon}(r):=e^{-r^{\varepsilon}}$, $\varepsilon > 0$, defined in polar coordinates. 
Evaluating $q^{\infty}_{\alpha,\sigma}[\varphi_{\varepsilon}]$ one performs the limit 
$\varepsilon \rightarrow 0$ to conclude that $q^{\infty}_{\alpha,\sigma}[\varphi_{\varepsilon}] < 0$ for 
small enough $\varepsilon$. The statement then follows by the variational principle \cite{BEH08}. Note 
that the factor of $2$ is due to the fact the there are two boundary segments of $\mathbb{R}^2_+$. 
\end{proof}
\begin{lemma} 
Assume that $\sigma,\alpha \in L^{\infty}(\mathbb{R}_+)$ have bounded support. Then there exist only 
finitely many negative eigenvalues.
\end{lemma}
\begin{proof} 
The statement follows from a bracketing argument, see \cite{BEH08} for a general discussion and 
\cite{KM16,KMBound} for applications of this technique. 

In a first step one writes 
$\mathbb{R}^2_+=B_R(\mathbb{R}^2_+) \cup \ (\mathbb{R}^2_+ \setminus B_R(\mathbb{R}^2_+))$ 
where $B_R(\mathbb{R}^2_+):=\{(x,y)\in \mathbb{R}^2_+: x^2+y^2 < R^2  \}$. The comparison operator 
then is a direct sum of two two-dimensional Laplacians, i.e., 
\begin{equation}
	-\Delta_{B_R(\mathbb{R}^2_+)} \oplus -\Delta_{\mathbb{R}^2_+ \setminus B_R(\mathbb{R}^2_+)} 
\end{equation}
with the same boundary conditions as $-\Delta^{d=\infty}_{\sigma,\alpha}$, except for additional 
Neumann boundary conditions along the dissecting line. We then choose $R$ large enough so that 
$\sigma=\alpha=0$ in $\mathbb{R}^2_+ \setminus B_R(\mathbb{R}^2_+)$. Accordingly, 
$-\Delta_{\mathbb{R}^2_+ \setminus B_R(\mathbb{R}^2_+)} $ is a positive operator. On the other hand, 
$-\Delta_{B_R(\mathbb{R}^2_+)}$ is defined on a bounded Lipschitz domain and hence has purely 
discrete spectrum, i.e., its essential spectrum is empty and there are only finitely many negative 
eigenvalues. Consequently, the operator bracketing
\begin{equation}
	-\Delta_{B_R(\mathbb{R}^2_+)} \oplus -\Delta_{\mathbb{R}^2_+ \setminus B_R(\mathbb{R}^2_+)}
	  \leq -\Delta^{d=\infty}_{\sigma,\alpha}
\end{equation}
implies the statement.
	%
	%
	%
	%
\end{proof}

\subsection{The model with hard-wall binding potential}
The model with non-vanishing binding potential, but vanishing contact interaction, was first studied 
in~\cite{KMBound}. The first important difference to the case where $v_{b}\equiv 0$ is that the 
two-particle configuration space is reduced from $\mathbb{R}^2_+$ to the `pencil-shaped' domain 
\begin{equation}
	\Omega:=\{(x,y)\in \mathbb{R}^2_+:\ |x-y| \leq d \}\ .
\end{equation}
Hence, the underlying Hilbert space is $L^2(\Omega)$ rather than $L^2(\mathbb{R}^2_+)$. As before, 
a rigorous realisation of \eqref{FormalHamiltonian} is obtained via the form
\begin{equation}
	q^{d}_{\alpha,\sigma}[\varphi]:=\int_{\Omega}|\nabla \varphi|^2 \ \ud x - \int_{0}^{d}\sigma(y)|
	\gamma(\varphi)|^2 \ud y + \int_{0}^{\infty}\alpha(y)|\varphi(y,y)|^2 \ \ud y\ ,
\end{equation}
which is defined on $\cD_{q}:=\{\varphi \in H^1(\Omega): \varphi|_{\partial \Omega_D}=0  \}$, where 
$\partial \Omega_D:=\{(x,y)\in \mathbb{R}^2_+:\ |x-y|=d  \}$. Note that the Dirichlet boundary conditions
along $\partial \Omega_D$ are due to the choice of the hard-wall binding potential. 
\begin{theorem} 
For every $\sigma, \alpha \in L^{\infty}(\mathbb{R}_+)$ the form 
$\left(q^{d}_{\alpha,\sigma},\cD_{q}\right)$ is bounded from below and is closed. 
\end{theorem}
As before, the representation theorem of forms assures the existence of a unique self-adjoint operator
associated with $q^{d}_{\alpha,\sigma}$ which shall be denoted by $-\Delta^{d}_{\sigma,\alpha}$. Again,
this operator acts as the standard two-dimensional Laplacian with coordinate dependent Robin boundary
conditions along the boundary segments with $x=0$ or $y=0$ as well as a jump condition along the
diagonal $x=y$ as before, see [Remark~1,\cite{KMBound}] for a more detail. 

So far the presence of a binding potential made no difference. However, as soon as we characterise 
the spectrum of $-\Delta^{d}_{\sigma,\alpha}$, the effect of the binding potential becomes obvious.
\begin{theorem}
\label{EssentialSpectrumII} 
Assume that $\sigma,\alpha \in L^{\infty}(\mathbb{R}_+)$ are given. Then 
$[2\pi^2/d^2,\infty) \subset \sigma_{ess}(-\Delta^{d}_{\sigma,\alpha})$. Furthermore, if 
$\alpha(y) \rightarrow 0$ as $y \rightarrow \infty$ one has 
$\sigma_{ess}(-\Delta^{d}_{\sigma,\alpha})=[\pi^2/2d^2,\infty)$.
\end{theorem}
\begin{proof} 
We only add some remarks, see [Theorem~2,\cite{KMBound}] and the proof of 
Theorem~\ref{TheoremEssentialSpectrumI} for more detail. 
	
In order to show the first part, one takes a Weyl sequence which consists of (normalised) ground states 
of the Dirichlet-Laplacian on rectangles $[0,L] \times [0,d/\sqrt{2}]$ that are placed in $\Omega$ such 
that the `L-boundary segments' touch $\partial \Omega_D$ as well as the diagonal $x=y$. The form 
$q^{d}_{\alpha,\sigma}$ evaluated for these states gives $2\pi^2/d^2+\pi^2/L^2$. Hence, letting $L$ tend
to any limit (including infinity) the statement follows. Note that the integral along the diagonal does not 
contribute due to the Dirichlet boundary conditions.
	
Regarding the second part one takes basically the same ground states but on rectangles 
$[0,L] \times [0,\sqrt{2}d]$. Now, the contribution of the integral along the diagonal does not vanish but,
since $\alpha(y) \rightarrow 0$ as $y \rightarrow \infty$, it can be made arbitrarily small. This proves the
statement.
\end{proof}
Theorem~\ref{EssentialSpectrumII} illustrates that, as long as the contact interaction strength converges 
to zero, the binding potential leads to a shift of the essential spectrum by at least $\pi^2/2d^2$. 
As for the effect on the discrete spectrum we first note that whenever $d=\infty$, by 
Theorem~\ref{TheoremEssentialSpectrumI} this is trivial for $\sigma=\alpha=0$. From a physical point 
of view this seems reasonable since there is no attractive potentials that could lead to bound states. 
However, quite surprisingly, for $d < \infty$ and $\sigma=\alpha=0$ we have the following 
result [Theorem~3,\cite{KMBound}].
\begin{theorem}
\label{TheoremDiscreteSpectrumPencil} 
Consider the self-adjoint operator $-\Delta^{d}_{\sigma=0,\alpha=0}$, i.e., we assume that
$\sigma=\alpha=0$. Then
\begin{equation}
		\sigma_{d}(-\Delta^{d}_{\sigma=0,\alpha=0}) \neq \emptyset\ .
\end{equation}
In other words, there exist eigenvalues below $\pi^2/2d^2$. 
\end{theorem}
Note that the existence of eigenvalues smaller than $\pi^2/2d^2$ for vanishing boundary and contact 
potential is a purely quantum mechanical effect. Furthermore, it is a geometrical effect since no 
non-trivial discrete spectrum would exist if one considered the two-particle system on the full line 
$\mathbb{R}$ instead of the half-line $\mathbb{R}_+$, see [Remark~4,\cite{KMBound}]. 

Of course, if one assumes that $\sigma(y) \geq 0$ for a.e. $y \in [0,d]$, and $\alpha(y)\leq 0$ for a.e. 
$y \in [0,\infty)$, the discrete spectrum will also be non-empty since the boundary integrals in 
$q^{d}_{\alpha,\sigma}$ are negative (note here the minus sign in the definition of the boundary potential 
$\sigma$). However, one may ask what happens when a positive boundary potential $\sigma$  becomes
large. Since this implies a strong repulsive singular two-particle interaction localised at the origin, 
bound states may no longer exist. Indeed, we have the following result.
\begin{theorem}\label{AbsenceDiscreteSpectrum} 
There exists a constant $\gamma < 0$ such that $\sigma_{d}(-\Delta^{d}_{\sigma,\alpha})=\emptyset$ 
whenever $\sigma(y) \leq \gamma$ for a.e. $y \in [0,d]$, $\alpha(y) \geq 0$ for a.e. $y \in [0,\infty)$ and 
$\alpha(y) \rightarrow 0$ as $y \rightarrow \infty$.
\end{theorem}
\begin{proof} 
Withouth contact potential $\alpha$ this result has been shown in [Theorem~4,\cite{KMBound}]. 
	
Now, by Theorem~\ref{EssentialSpectrumII} we conclude that 
$\inf \sigma_{ess}(-\Delta^{d}_{\sigma,\alpha})=\pi^2/2d^2$. Furthermore, since $\alpha$ is assumed 
to be strictly positive, the corresponding operator is larger (in the sense of an operator bracketing) than 
the operator with same boundary potential $\sigma$, but without contact potential. Consequently, if there
existed an eigenvalue smaller than $\pi^2/2d^2$ the same would hold for the operator without contact
potential. This, however, is in contradiction with [Theorem~4,\cite{KMBound}].
\end{proof}
Theorem~\ref{AbsenceDiscreteSpectrum} shows that strong singular interactions at the origin (without
contact interaction) destabilise the system in the sense that no discrete spectrum is present anymore
when compared to the free system with $\sigma=\alpha=0$.

\subsection{Random singular pair interactions}
In this subsection we consider a generalisation of the Hamiltonian~\eqref{FormalHamiltonian} in the sense
that the singular, vertex-induced pair interactions are not only present in the origin or the vertex 
of the graph. This seems desirable since, as described previously, localised two-particle interactions can 
be associated with defects in the metal and such defects occur, of course, not only at the origin but 
everywhere in the wire. We note that this model was formulated in~\cite{KernerRandomII} to which we 
also refer for more detail. 

Since the spatial positions of defects in a real metal varies from metal to metal it seems reasonable not 
to work with a specific (deterministic) two-particle Hamiltonian, but with a random one. In other words, in
this section we enter the realm of random Schr\"{o}dinger operators which in recent years have become 
an important research area \cite{stollmann2001caught,KirschInvitation}. Most importantly, using the 
language of random Schr\"{o}dinger operators, one has been able to give a rigorous description of 
various phenomena in physics such as Anderson 
localisation~\cite{AndersonLocalisatioin,SimonBookSchrödinger}.

Turning to our model, we consider a system of two particles on the half-line $\mathbb{R}_+$ whose
random Hamiltonian shall formally be given by 
\begin{equation}\label{RandomHamiltonian}
	H_{\omega}=-\frac{\partial^2}{\partial x^2}-\frac{\partial^2}{\partial y^2}+v_b(|x-y|)+\sum_{i=1}^{\infty}
	v_{i}(x,y)\left[\delta(x-a_i(\omega))+\delta(y-a_i(\omega)) \right]\ ,
\end{equation}
where $(a_i(\omega))_{i \in \mathbb{N}}$ are the random positions of the defects, called atoms in the 
sequel. As before we assume that $(v_i)_{i \in \mathbb{N}}$ are real-valued, bounded and symmetric,
$v_i(x,y)=v_i(y,x)$. Furthermore, we define $l_i(\omega):=a_{i}(\omega)-a_{i-1}(\omega)$ for $i \geq 1$
and set $a_0(\omega):=0$. In other words, $l_i(\omega)$ is the random distance between the $i-1$-st 
and the $i$-th atom. 

Now we consider the lengths $(l_i(\omega))_{i \in \mathbb{N}}$ as a family of independent random
variables over some probability space $(\Pi,\xi,\mathbb{P})$ generated by a Poisson process, see 
\cite{StolzPoisson} for more detail. More explicitly, we assume that the probability for the length 
$l_i(\omega)$ to be in the interval $[a,b]$ is given by
\begin{equation}
	\mathbb{P}\left[l_i \in [a,b] \right]=\nu \int_{a}^{b}\ue^{-\nu l} \ud l \ ,
\end{equation}
with $\nu  > 0$ denoting the Poisson density. 

Again, due to the presence of $\delta$-potentials in~\eqref{RandomHamiltonian} we shall again use a
suitable quadratic form to rigorously construct a self-adjoint operator that is associated with the formal 
expression~\eqref{RandomHamiltonian}. We introduce
\begin{equation}\label{RandomForm}
	q_{\omega}[\varphi]=\int_{\Omega}|\nabla \varphi|^2 \ \ud x+ \sum_{i=1}^{\infty}
	\int_{\Gamma_i(\omega)}\sigma_i(y)|\gamma_{i}(\varphi)|^2 \ \ud y\ ,
\end{equation}
where $\gamma_{i}(\varphi)$ denotes the restriction (in the sense of traces of Sobolev functions) to 
\begin{equation}
	\Gamma_i(\omega):=\{(x,y)\in \Omega: x=a_i(\omega) \ \text{or}\ y=a_i(\omega)  \}\ .
\end{equation}
Furthermore, we set $\sigma_i(y):=-v_{i}(0,y)$. Due to the infinite sum appearing in~\eqref{RandomForm}
it may not be possible to define $q_{\omega}$ on all of $\cD_{q} \subset H^1(\Omega)$. Since we want to
find a closed realisation of the form $q_{\omega}$ we have to guess a suitable sub-domain. Indeed, one
has the following result [Theorem~2.1,\cite{KernerRandomII}]. 
\begin{theorem} 
Let $(\sigma_i(\omega))_{i \in \mathbb{N}} \subset L^{\infty}(\mathbb{R}_+)$ be given. Then the form 
$q_{\omega}$ on the domain 
\begin{equation}
		\cD_{q}(\omega)=\{\varphi \in H^1(\Omega): q_{\omega}[\varphi] < \infty \}
\end{equation}
is positive and closed for almost every $\omega \in \Pi$. 
\end{theorem}
We denote the unique self-adjoint operator associated with the form $q_{\omega}$ as 
$-\Delta_{\sigma}(\omega)$.

The random Schr\"{o}dinger operators usually considered in the literature have an astonishing property, 
namely that the spectrum is almost surely non-random \cite{PasturFigotin,KirschInvitation}, which is due 
to a certain ergodicity property of the models. For our model, we will see that only the essential part of 
the spectrum is non-random. The discrete part, however, is random. Indeed we have the following 
results [Theorem~3.1, Lemma~3.4, Theorem~3.5,\cite{KernerRandomII}]. 
\begin{theorem}
Let $(\sigma_i(\omega))_{i \in \mathbb{N}} \subset L^{\infty}(\mathbb{R}_+)$ be given. Then
\begin{equation}
		\sigma_{ess}(-\Delta_{\sigma}(\omega))=[\pi^2/2d^2,\infty)
\end{equation}
almost surely. 
\end{theorem}
The discrete part of the spectrum, on the other hand, is random. More explicitly, we obtain the following result.
\begin{theorem}
\label{TheoremAbsenceDiscreteSpectrum} 
Let $(\sigma_i(\omega))_{i \in \mathbb{N}} \subset L^{\infty}(\mathbb{R}_+)$ be given. Then 
\begin{equation}
		\mathbb{P}[\sigma_{d}(-\Delta_{\sigma}(\omega))\neq\emptyset] > 0\ .
\end{equation}
Furthermore, there exists a constant $\gamma=\gamma(d) > 0$ such that if $\inf \sigma_k > \gamma$ 
for one $k\in \mathbb{N}$ then
\begin{equation}
		\mathbb{P}[\sigma_{d}(-\Delta_{\sigma}(\omega))=\emptyset] > 0\ .
\end{equation}
\end{theorem}
Theorem~\ref{TheoremAbsenceDiscreteSpectrum} tells us that the discrete part of the spectrum is
destroyed with finite probability as well as conserved with finite probability. This leads to an interesting
physical implication: In general, disorder is associated with a suppression of transport as in the Anderson
localisation phenomenon. However, assuming that no dense pure point spectrum is created in 
$[\pi^2/2d^2,\infty)$ and that the density of states does not change, 
Theorem~\ref{TheoremAbsenceDiscreteSpectrum} implies that disorder may lead to an improvement of
transport with finite probability according to `Fermi's golden rule'.  

\subsection{The condensation of electron pairs in a quantum wire}\label{SectionCondensationPairs}
In this subsection we want to report on the results that were obtained in
\cite{KernerElectronPairs,KernerSurfaceDefects,KernerInteractingPairs}. Since we are interested in pairs
of particles, we consider the case where $d< \infty$, i.e., we assume that a hard-wall binding potential
$v_{b}$ is present. In the previous sections we worked on the full Hilbert space $L^2(\Omega)$ describing
two distinguishable and spinless particles. However, since we are interested in applying the 
Hamiltonian~\eqref{FormalHamiltonian} to understand superconductivity, which involves electrons, we
need to implement the exchange symmetry of identical particles. 

In this review we restrict ourselves to the case considered in~\cite{KernerElectronPairs} where the two
electrons are assumed to have the same spin. This leads to the requirement that the two-particle wave
function has to be anti-symmetric. The case of opposite spin, which is realised in actual Cooper
pairs, is considered in~\cite{KernerInteractingPairs}. We only mention here that the results regarding 
the condensation there are comparable. 

In order to ensure anti-symmetry of the wave function we work in the anti-symmetric subspace
\begin{equation}
	L^2_{a}(\Omega):=\{\varphi \in L^2(\Omega):\ \varphi(x,y)=-\varphi(y,x)   \}\ .
\end{equation}
We then introduce the quadratic form 
\begin{equation}\label{FormCondensation}
	q^{d}_{\sigma}[\varphi]:=\frac{\hbar^2}{2m_e}\int_{\Omega}|\nabla \varphi|^2 \ \ud x - \int_{0}^{d}\sigma(y)|\gamma(\varphi)|^2 \ud y \ ,
\end{equation}
where $\sigma \in L^{\infty}(\mathbb{R}_+)$, on this subspace. Here we added physical constants with
$m_e$ denoting the electron mass. The domain of the form is given by 
$\cD_{q}:=\{\varphi \in H^1(\Omega)\cap L^2_{a}(\Omega): \ \varphi|_{\partial \Omega_D}=0  \}$. Again, 
this form is closed and bounded from below, and hence there exists a unique self-adjoint operator 
associated with this form. This is the Hamiltonian of our two-particle system. We denote this operator,
which again acts as the standard two-dimensional Laplacian, as $-\Delta^{d}_{\sigma}$. 

\begin{theorem}[\cite{KernerElectronPairs}]\label{TheoremBoundElectrons} 
One has 
\begin{equation}
		\sigma_{ess}(-\Delta^{d}_{\sigma})=[\hbar^2 \pi^2 /m_e d^2,\infty)\ .
\end{equation}
Furthermore, if $\sigma=0$ then
\begin{equation}
		\sigma_{d}(-\Delta^{d}_{\sigma=0})=\{E_0 \}\ ,
\end{equation}
i.e., there is exactly one eigenvalue with multiplicity one below the bottom of the essential spectrum. 
In addition, one has
\begin{equation}
		0.25 \cdot \frac{\hbar^2 \pi^2}{m_ed^2}\leq E_0 \leq 0.93 \cdot \frac{\hbar^2 \pi^2}{m_ed^2}\ .
\end{equation}
\end{theorem}
Theorem~\ref{TheoremBoundElectrons} has an interesting physical consequence: one important
measurable quantity associated with the superconducting phase of a metal is the so-called spectral gap 
$\Delta > 0$, see~\cite{MR04}. This spectral gap is responsible, for example, for the exponential decay of
the specific heat at temperatures lower than the critical one. It is one of the successes of the BCS-theory
that the spectral gap can be interpreted as the binding energy of a single Cooper pair. In other words, the
spectral gap measures the energy necessary to break up one Cooper pair. Due to the choice of the
hard-wall binding potential, in our model the pair cannot be broken up. However, it is possible to excite a 
pair. Since, as we will see later, the pairs condense into the ground state it seems reasonable to identify
the spectral gap as the excitation of a pair from the ground state to the first excited states. In other words,
in our model one obtains the relation 
\begin{equation}\label{SpectralGap}
	\Delta=\Delta(d) \sim \frac{\hbar^2 \pi^2}{m_ed^2}
\end{equation}
for the spectral gap. This relation establishes a direct link between the spatial extension of a pair and the
spectral gap. In particular, since the spectral gap in superconducting metals is of order $10^{-3}$eV 
\cite{MR04}, the relation~\eqref{SpectralGap} implies that $d$ is of the order $10^{-6}$m.  
Interestingly, this agrees with Cooper's estimate as presented in \cite{CooperBoundElectron}.

In order to study the condensation phenomenon (similar to BEC as in Section~\ref{secBEC}) of electron
pairs one has to employ methods from quantum statistical mechanics (see, e.g., \cite{SchwablSM}). 
In particular, one has to perform a thermodynamic limit as in Section~\ref{secBEC}, and this requires 
to restrict the system from the half-line to the interval $[0,L]$. The underlying Hilbert space then is 
$L^2_{a}(\Omega_L)$, with
\begin{equation}
	\Omega_L:=\{(x,y) \in \Omega: 0 \leq x,y \leq L\}\ .
\end{equation}
The natural generalisation of~\eqref{FormCondensation} is defined on the domain 
$\cD_{q_{L}}:=\{\varphi \in H^1(\Omega_L)\cap L^2_{a}(\Omega_L):\ \varphi|_{\partial \Omega_{L,D}}=0\}$
with $\partial \Omega_{L,D}:=\{(x,y) \in \partial \Omega_L:\ |x-y|=d \ \text{or}\ x=L \ \text{or}\  y=L \}$. In 
other words, one introduces additional Dirichlet boundary conditions along the dissecting lines $x=L$ and 
$y=L$. We denote the associated self-adjoint operator by $-\Delta^{d}_{\sigma,L}$. 

Since $\Omega_L$ is a bounded Lipschitz domain, $-\Delta^{d}_{\sigma,L}$ has purely discrete spectrum. 
We denote its corresponding eigenvalues, counted with multiplicity, by 
$\{E^{\sigma}_n(L) \}_{n\in \mathbb{N}_0}$.
\begin{lemma}\label{ConvergenceEigenvalues} 
Assume that $\sigma=0$. Then 
\begin{equation}
		\lim_{L \rightarrow \infty}E^{\sigma=0}_0(L)=E_0\ .
\end{equation}
Furthermore, $E^{\sigma=0}_n(L) \geq \frac{\hbar^2\pi^2}{m_ed^2}$ for all $n \geq 1$ and $L > d$.
\end{lemma} 
Lemma~\ref{ConvergenceEigenvalues} implies the existence of a finite spectral gap in the thermodynamic limit which eventually is
responsible for the condensation of the pairs. 

Recalling Definition~\ref{BECdef}, we can now establish the main result of this section. 
\begin{theorem}\label{CondensationPairsTheorem} 
For $\sigma=0$ there exists a critical density $\rho_{crit}(\beta)$ such that the ground state is
macroscopically occupied in the thermodynamic limit for all pair densities $\rho > \rho_{crit}(\beta)$.
Furthermore, there exists a constant $\gamma < 0$ such that, for all pair densities $\rho > 0$, no 
eigenstate is macroscopically occupied if $\|\sigma\|_{\infty} < \gamma$. 
\end{theorem}
\begin{proof} For the proof see the proofs of [Theorem~3.3 and Theorem~3.6,\cite{KernerElectronPairs}]
as well as [Theorem~4.4,\cite{KernerInteractingPairs}]. 
\end{proof}
Theorem~\ref{CondensationPairsTheorem} shows that the pairs condense in the quantum wire given 
that there are no repulsive singular two-particle interactions localised at the origin. However, if the 
singular interactions are strong enough, the condensate in the ground state will be destroyed. Hence, 
if one identifies the superconducting phase with the presence of a condensate of pairs (here in an eigenstate 
for non-interacting pairs), Theorem~\ref{CondensationPairsTheorem} shows that the superconducting
phase in a quantum wire can be destroyed by singular two-particle interactions.

\subsection{The impact of surface defects on the superconducting phase}
In this final section we report on yet another application of the two-particle model introduced above which
was presented in~\cite{KernerSurfaceDefects}. More explicitly, we extend the model characterised by the
form~\eqref{FormCondensation} and the associated Hamiltonian $-\Delta^{d}_{\sigma}$ defined on the
anti-symmetric Hilbert space $L^2_{a}(\Omega)$. However, we will only consider the case where there 
are no singular, vertex-induced two-particle interactions at the origin, i.e., we set $\sigma=0$. 

In Section~\ref{SectionCondensationPairs} we investigated the (Bose-Einstein) condensation of pairs of
electrons. Theorem~\ref{CondensationPairsTheorem} shows that the pairs condense into the ground state
if no singular interactions are present and given the pair density $\rho > 0$ is large enough. Also, the 
presence of condensation is paramount for the existence of the superconducting phase. Real metals
are never perfect and there exist defects that affect the behaviour of electrons in the bulk. However,
besides defects in the bulk, a real metal will also exhibit defects on the surface, i.e., a real surface will not
be arbitrarily smooth. Note that the existence of a surface is, to a first approximation, not taken into
account in most discussions in solid state physics, since the solid is modelled to be infinitely extended 
in order to conserve periodicity. However, it has also long become clear that surface effects cannot be
neglected altogether~\cite{FossheimSuperconducting}. It is aim of this section to introduce a model to
investigate the effect of surface defects on the superconducting phase in the bulk of a quantum wire by 
investigating their effect on the condensation of electron pairs in the bulk. 

In order to take surface defects into account we have to extend our Hilbert space. More explicitly, we set 
\begin{equation}\label{HilbertSpaceSurface}
	\cH:=L^2_{a}(\Omega) \oplus \ell^2(\mathbb{N})\ ,
\end{equation}
where $\ell^2(\mathbb{N})$ is the space of square-summable sequences. Consequently, a given pair of 
electrons is described by a state of the form $\left(\varphi,f\right)^T$, with $\varphi \in L^2_{a}(\Omega)$
and $f \in \ell^2(\mathbb{N})$. This means that we model the surface defects as the vertices of the discrete
graph $\mathbb{N}$, which seems reasonable in a regime where the spatial extension of those defects 
is small compared to the bulk. 

The Hamiltonian of a free pair of electrons on $\cH$ is given by
\begin{equation}
	H_{p}:=-\Delta^{d}_{\sigma=0} \oplus \mathcal{L}(\gamma)\ ,
\end{equation}
where $\mathcal{L}(\gamma)$ is the (weighted) discrete Laplacian acting via
\begin{equation}
	(\mathcal{L}(\gamma)f)(n)=\sum_{m=1}^{\infty}\gamma_{mn}\left(f(m)-f(n)\right)\ ,
\end{equation}
with $(\gamma)_{mn}=:\gamma_{mn}=\delta_{|n-m|,1}e_n$ and $(e_n)_{n \in \mathbb{N}} \subset \mathbb{R}_+$. 

Now, since we are interested in the condensation phenomenon we have to restrict the system to a finite
volume as we have done in the previous section. More explicitly, the finite volume Hilbert space is given by 
\begin{equation}
	\cH_L:=L^2_{a}(\Omega_L) \oplus \mathbb{C}^{n(L)}\ ,
\end{equation} 
where $n(L) \in \mathbb{N}$ denotes the number of surface defects in the interval $[0,L]$. On this Hilbert
space one considers $H^L_p$, i.e., the restriction of $H_p$ to the finite-volume Hilbert space $\cH_L$. 
This operator has purely discrete spectrum and the eigenvalues are the union of those coming from 
$-\Delta^{d}_{\sigma}|_{L^2_{a}(\Omega_L)}$ (where this operator is defined as in the previous section)
and $\mathcal{L}(\gamma)|_{\mathbb{C}^{n(L)}}$. 

In order to formulate the model it is convenient to use the formalism of second quantisation~\cite{MR04}.
This means that one works on the Fock space over $\cH_L$, rather than on $\cH_L$ itself. The second 
quantisation of $H^L_p$ is given by
\begin{equation}\label{SecondQuantisationFreeHamiltonian}
	\Gamma(H^L_p)=\sum_{n=0}^{\infty}E^{\sigma=0}_{n}(L)a^{\ast}_{n}a_{n}+\sum_{k=1}^{n(L)}\lambda_k(L)b^{\ast}_{k}b_{k}\ ,
\end{equation}
where $(E^{\sigma=0}_{n}(L))_{n \in \mathbb{N}_0}$ are the eigenvalues of 
$-\Delta^{d}_{\sigma=0}|_{L^2_{a}(\Omega_L)}$ and $(\lambda_k(L))_{k=1,...,n(L)}$ are the eigenvalues
of $\mathcal{L}(\gamma)|_{\mathbb{C}^{n(L)}}$, counted with multiplicity. Furthermore, $(a^{\ast}_n,a_n)$
are the creation and annihilation operators of the states $\varphi_n \oplus 0$, where
$\varphi_n \in L^2_a(\Omega_L)$ are the corresponding eigenstate of 
$-\Delta^{d}_{\sigma=0}|_{L^2_{a}(\Omega_L)}$. In contrast, $(b^{\ast}_{k},b_{k})$ are the creation 
and annihilation operators of the states $0 \oplus f_n$, where $f_n \in \mathbb{C}^{n(L)}$ are the 
corresponding eigenstates of $\mathcal{L}(\gamma)|_{\mathbb{C}^{n(L)}}$. To obtain the full Hamiltonian
of the model we extend the free Hamiltonian~\eqref{SecondQuantisationFreeHamiltonian} and write
\begin{equation}\label{ModelHamiltonianSurface}
	H_L(\rho_s,\alpha,\lambda)=\Gamma(H^L_p)-\alpha\sum_{k=1}^{n(L)}b^{\ast}_{k}b_{k}+
	\lambda\rho_s(\mu_L,L)\sum_{k=1}^{n(L)}b^{\ast}_{k}b_{k}\ ,
\end{equation}
where $\alpha \geq 0$ describes the surface tension; $\lambda \geq 0$ is an interaction strength
associated with the repulsion of the pairs in the surface defects and $\rho_s(\mu_L,L)$ is the density of 
pairs on $\mathbb{C}^{n(L)}$, see the equation below. Note here that 
$\sum_{k=1}^{n(L)}b^{\ast}_{k}b_{k}$ is the (surface-) number operator whose expectation value equals 
the number of pairs in the surface defects. Also, the third term in~\eqref{ModelHamiltonianSurface} is
added to take into account repulsive interactions between electron pairs accumulating in the surface
defects which are expected since the surface defects are imagined to be relatively small. The explicit form
of this term follows from a simplification of standard mean-field considerations where the interaction term
is generally of the form $\lambda \hat{N}^2/V$, where $\hat{N}$ is the number operator and $V$ is the 
volume of the system. In other words, we have replaced $\hat{N}/V$ by the density $\rho_s(\mu_L,L)$ 
for which
\begin{equation}
	\rho_s(\mu_L,L):=\frac{\omega^{H_L(\rho_s,\alpha,\lambda)}_{\beta,\mu_L}
	\left(\sum_{k=1}^{n(L)}b^{\ast}_{k}b_{k}\right)}{n(L)}
\end{equation}
holds with $\omega^{H_L(\rho_s,\alpha,\lambda)}_{\beta,\mu_L}(\cdot)$ denoting the Gibbs state of the
grand-canonical ensemble at inverse temperature $\beta=1/T$ and chemical potential $\mu_L$. 

The advantage of the Hamiltonian $H(\rho_s,\alpha,\lambda)$ is that it can be rewritten as
\begin{equation}\label{RewrittenHamiltonian}
	H_L(\rho_s,\alpha,\lambda)=\sum_{n=0}^{\infty}E^{\sigma=0}_{n}(L)a^{\ast}_{n}a_{n}+
	\sum_{k=1}^{n(L)}\left(\lambda_k(L)+\lambda\rho_s(\mu_L,L)-\alpha\right) b^{\ast}_{k}b_{k} \ ,
\end{equation}
which yields an effective, non-interacting many-pair model with shifted eigenvalues for the discrete part.
Note that, in particular, \eqref{RewrittenHamiltonian} implies 
$\mu_L < \min\{\lambda\rho_s(\mu_L,L)-\alpha,E_0(L)  \}$, taking into account that $\lambda_1(L)=0$. 

The goal then is to investigate the macroscopic occupation of the ground state $\varphi_0 \oplus 0$ in a suitable
thermodynamic limit (see~\cite{KernerSurfaceDefects} for details) for the 
Hamiltonian~\eqref{RewrittenHamiltonian}. It turns out that a key quantity is the inverse density of surface
defects $\delta > 0$ defined as
\begin{equation}
	\delta:=\lim_{L \rightarrow \infty}\frac{L}{n(L)}\ .
\end{equation}
One obtains the following result.
\begin{theorem}\label{DestructionSurface} 
If 
\begin{equation}\label{ConditionSurfaceI}
		2\lambda \cdot \delta\cdot \rho < E_0+\alpha
\end{equation}
holds, no eigenstate $\varphi_n \oplus 0$ is macroscopically occupied in the thermodynamic limit. This means, in
particular, that the condensate in the bulk is destroyed for arbitrary pair densities whenever $\lambda=0$.
\end{theorem}
Theorem~\ref{DestructionSurface} has the remarkable consequence that the condensation in the bulk is
destroyed for all pair densities $\rho > 0$ in the following cases: the pairs do not repel each other which
allows them to accumulate in the surface defects or the number of surface defects is very large. 

Finally, we obtain the following result.
\begin{theorem}\label{MainTheorem}
Assume that $\delta,\lambda > 0$. Then there exists a critical pair density $\rho_{crit}=\rho_{crit}(\beta,\alpha,\lambda) > 0$
such that for all pair densities $\rho > \rho_{crit}$ the state $\varphi_0 \oplus 0$ is macroscopically 
occupied in the thermodynamic limit. 
\end{theorem}
Theorem~\ref{MainTheorem} shows that the superconducting phase in the bulk can be recovered given the interaction strength $\lambda > 0$ is non-zero and, most importantly, given the number of surface impurities is not too large.

\vspace*{0.5cm}

{\small
\bibliographystyle{amsalpha}
\bibliography{Literature}}

\newcommand{\etalchar}[1]{$^{#1}$}
\def\cprime{$'$} \def\polhk#1{\setbox0=\hbox{#1}{\ooalign{\hidewidth
  \lower1.5ex\hbox{`}\hidewidth\crcr\unhbox0}}} \def\cprime{$'$}
  \def\polhk#1{\setbox0=\hbox{#1}{\ooalign{\hidewidth
  \lower1.5ex\hbox{`}\hidewidth\crcr\unhbox0}}} \def\cprime{$'$}
  \def\polhk#1{\setbox0=\hbox{#1}{\ooalign{\hidewidth
  \lower1.5ex\hbox{`}\hidewidth\crcr\unhbox0}}} \def\cprime{$'$}
  \def\polhk#1{\setbox0=\hbox{#1}{\ooalign{\hidewidth
  \lower1.5ex\hbox{`}\hidewidth\crcr\unhbox0}}}
\providecommand{\bysame}{\leavevmode\hbox to3em{\hrulefill}\thinspace}
\providecommand{\MR}{\relax\ifhmode\unskip\space\fi MR }
\providecommand{\MRhref}[2]{%
  \href{http://www.ams.org/mathscinet-getitem?mr=#1}{#2}
}
\providecommand{\href}[2]{#2}
\begin{thebibliography}{EKK{\etalchar{+}}08}

\bibitem[And58]{AndersonLocalisatioin}
P.~W. Anderson, \emph{Absence of diffusion in certain random lattices}, Phys.
  Rev. \textbf{109} (1958), 1492--1505.

\bibitem[B\'80]{Ber80}
P.H. B\'erard, \emph{Spectres et groupes cristallographiques. {I}. {D}omaines
  euclidiens}, Invent. Math. \textbf{58} (1980), no.~2, 179--199.

\bibitem[BCS57]{BCSI}
J.~Bardeen, L.~N. Cooper, and J.~R. Schrieffer, \emph{Theory of
  superconductivity}, Phys. Rev. \textbf{108} (1957), 1175--1204.

\bibitem[BE09]{BolEnd09}
J.~Bolte and S.~Endres, \emph{The trace formula for quantum graphs with general
  self-adjoint boundary conditions}, Ann. H. Poincare \textbf{10} (2009),
  189--223.

\bibitem[BEKS94]{BrascheExnerKuperinSeba92}
J.~F. Brasche, P.~Exner, Y.~A. Kuperin, and P.~Seba,
  \emph{Schr\"{o}dinger-operators with singular interactions}, Journal of
  Mathematical Analysis and Applications \textbf{184} (1994), no.~1, 112 --
  139.

\bibitem[BER15]{BolEggRue15}
J.~Bolte, S.~Egger, and R.~Rueckriemen, \emph{Heat-kernel and resolvent
  asymptotics for {S}chr\"odinger operators on metric graphs}, Appl. Math. Res.
  Express. AMRX (2015), no.~1, 129--165.

\bibitem[BES15]{BolEggSte15}
J.~Bolte, S.~Egger, and F.~Steiner, \emph{Zero modes of quantum graph
  {L}aplacians and an index theorem}, Ann. Henri Poincar\'e \textbf{16} (2015),
  no.~5, 1155--1189.

\bibitem[Bet31]{Bet31}
H.~Bethe, \emph{{Z}ur {T}heorie der {M}etalle {I}.}, Z. Phys. \textbf{71}
  (1931), 205--226.

\bibitem[BG17a]{BolGar17}
J.~Bolte and G.~Garforth, \emph{Exactly solvable interacting two-particle
  quantum graphs}, J. Phys. A \textbf{50} (2017), no.~10, 105101, 27.

\bibitem[BG17b]{BolGar17a}
\bysame, \emph{Solvable models of interacting $n$-particle models on quantum
  graphs}, arXiv:1704.00469, 2017.

\bibitem[BHE08]{BEH08}
J.~Blank, M.~Havli\v{c}ek, and P.~Exner, \emph{Hilbert space operators in
  quantum physics}, Springer Verlag, 2008.

\bibitem[BK13a]{Berkolaiko:2013}
G.~Berkolaiko and P.~Kuchment, \emph{Introduction to quantum graphs}, American
  Mathematical Society, Providence, RI, 2013.

\bibitem[BK13b]{BKSingular}
J.~Bolte and J.~Kerner, \emph{Quantum graphs with singular two-particle
  interactions}, Journal of Physics A: Mathematical and Theoretical \textbf{46}
  (2013), no.~4, 045206.

\bibitem[BK13c]{BKContact}
\bysame, \emph{Quantum graphs with two-particle contact interactions}, Journal
  of Physics A: Mathematical and Theoretical \textbf{46} (2013), no.~4, 045207.

\bibitem[BK14]{BolteKernerBEC}
\bysame, \emph{Many-particle quantum graphs and {B}ose-{E}instein
  condensation}, Journal of Mathematical Physics \textbf{55} (2014), no.~6.

\bibitem[BK16]{BolKer16}
J.~Bolte and J.~Kerner, \emph{Instability of {B}ose-{E}instein condensation
  into the one-particle ground state on quantum graphs under repulsive
  perturbations}, J. Math. Phys. \textbf{57} (2016), no.~4, 043301, 9.

\bibitem[BL10]{BL10}
J.~Behrndt and A.~Luger, \emph{On the number of negative eigenvalues of the
  {L}aplacian on a metric graph}, Journal of Physics A \textbf{43} (2010),
  1--10.

\bibitem[BM06]{BelMin06}
B.~Bellazzini and M.~Mintchev, \emph{Quantum fields on star graphs}, J. Phys. A
  \textbf{39} (2006), no.~35, 11101--11117.

\bibitem[BMLL13]{Behrndt2013}
J.~Behrndt, Matthias M.~Langer, and V.~Lotoreichik, \emph{Schr{\"o}dinger
  operators with $\delta$ and $\delta^{\prime}$-potentials supported on
  hypersurfaces}, Annales Henri Poincar{\'e} \textbf{14} (2013), no.~2,
  385--423.

\bibitem[Cau15]{Cau15}
V.~Caudrelier, \emph{On the inverse scattering method for integrable {PDE}s on
  a star graph}, Comm. Math. Phys. \textbf{338} (2015), no.~2, 893--917.

\bibitem[CC07]{CauCra07}
V.~Caudrelier and N.~Cramp\'e, \emph{Exact results for the one-dimensional
  many-body problem with contact interaction: including a tunable impurity},
  Rev. Math. Phys. \textbf{19} (2007), no.~4, 349--370.

\bibitem[CCG{\etalchar{+}}11]{cazalilla2011one}
M.A. Cazalilla, R.~Citro, T.~Giamarchi, E.~Orignac, and M.~Rigol, \emph{One
  dimensional bosons: From condensed matter systems to ultracold gases},
  Reviews of Modern Physics \textbf{83} (2011), no.~4, 1405.

\bibitem[CFKS87]{SimonBookSchrödinger}
H.~L. Cycon, R.~G. Froese, W.~Kirsch, and B.~Simon, \emph{Schr\"odinger
  operators with application to quantum mechanics and global geometry}, Texts
  and Monographs in Physics, Springer-Verlag, Berlin, 1987.

\bibitem[Coo56]{CooperBoundElectron}
Leon~N. Cooper, \emph{Bound electron pairs in a degenerate {F}ermi gas}, Phys.
  Rev. \textbf{104} (1956), 1189--1190.

\bibitem[Dob05]{Dob05}
M.~Dobrowolski, \emph{{Angewandte Funktionalanalysis: Funktionalanalysis,
  Sobolev-R\"aume und Elliptische Differentialgleichungen}}, {Springer-Verlag},
  {Berlin}, 2005.

\bibitem[Ein25]{EinsteinBEC}
A.~Einstein, Sitzber. {K}gl. {P}reuss. {A}kadm. {W}iss. (1925), 3.

\bibitem[EK17]{EggerKerner}
S.~Egger and J.~Kerner, \emph{Scattering properties of two singularly
  interacting particles on the half-line}, Rev. Math. Phys. \textbf{29} (2017),
  no.~10.

\bibitem[EKK{\etalchar{+}}08]{Exnetal08}
P.~Exner, J.~P. Keating, P.~Kuchment, T.~Sunada, and A.~Teplyaev (eds.),
  \emph{Analysis on graphs and its applications}, Proceedings of Symposia in
  Pure Mathematics, vol.~77, American Mathematical Society, Providence, RI,
  2008.

\bibitem[FS04]{FossheimSuperconducting}
K.~Fossheim and A.~Sudboe, \emph{Superconductivity: Physics and applications},
  Wiley, 2004.

\bibitem[Gau71]{Gau71}
M.~Gaudin, \emph{Boundary energy of a {B}ose gase in one dinmension}, Phys.
  Rev. A \textbf{4} (1971), 386--394.

\bibitem[Gau14]{Gau14}
\bysame, \emph{The {B}ethe wavefunction}, Cambridge University Press, New York,
  2014, Translated from the 1983 French original by Jean-S\'ebastien Caux.

\bibitem[GG08]{GG08}
N.~Goldman and P.~Gaspard, \emph{Quantum graphs and the integer quantum {H}all
  effect}, Phys. Rev. B \textbf{77} (2008), 024302.

\bibitem[Gir60]{G60}
M.~D. Girardeau, \emph{Relationship between systems of impenetrable bosons and
  fermions in one dimension}, Journal of Mathematical Physics \textbf{1}
  (1960), 516--523.

\bibitem[Gla93]{glasser1993solvable}
M.~L. Glasser, \emph{A solvable non-separable quantum two-body problem},
  Journal of Physics A: Mathematical and General \textbf{26} (1993), no.~17,
  L825.

\bibitem[GN05]{glasser2005solvable}
M.~L. Glasser and L.~M. Nieto, \emph{Solvable quantum two-body problem:
  entanglement}, Journal of Physics A: Mathematical and General \textbf{38}
  (2005), no.~24, L455.

\bibitem[GS06]{GNUSMY06}
S.~Gnutzmann and U.~Smilansky, \emph{Quantum graphs: Applications to quantum
  chaos and universal spectral statistics}, Taylor and Francis. Advances in
  Physics \textbf{55} (2006), 527--625.

\bibitem[Gut90]{Gut90}
M.C. Gutzwiller, \emph{Chaos in classical and quantum mechanics},
  Interdisciplinary Applied Mathematics, vol.~1, Springer-Verlag, New York,
  1990.

\bibitem[Har07]{Ha07}
M.~Harmer, \emph{Two particles on a star graph, {I}}, Russian Journal of
  Mathematical Physics \textbf{14} (2007), 435--439.

\bibitem[Har08]{Ha08}
\bysame, \emph{Two particles on a star graph, {II}}, Russian Journal of
  Mathematical Physics \textbf{15} (2008), 473--480.

\bibitem[HKR11]{HarKeaRob11}
J.~M. Harrison, J.~P. Keating, and J.~M. Robbins, \emph{Quantum statistics on
  graphs}, Proc. R. Soc. Lond. Ser. A Math. Phys. Eng. Sci. \textbf{467}
  (2011), no.~2125, 212--233.

\bibitem[HKRS14]{HarKeaRobSaw14}
J.~M. Harrison, J.~P. Keating, J.~M. Robbins, and A.~Sawicki,
  \emph{{$n$}-particle quantum statistics on graphs}, Comm. Math. Phys.
  \textbf{330} (2014), no.~3, 1293--1326.

\bibitem[Hoh67]{Hohenberg}
P.~C. Hohenberg, \emph{Existence of long-range order in one and two
  dimensions}, Phys. Rev. \textbf{158} (1967), 383--386.

\bibitem[HV16]{QuantumWellsWiresDots}
P.~Harrison and A.~Valavanis, \emph{Quantum wells, wires and dots}, Wiley,
  2016.

\bibitem[Kera]{KernerRandomI}
J.~Kerner, \emph{{A remark on the effect of random singular two-particle
  interactions }}, preprint, \url{https://arxiv.org/abs/1707.00961}, 2017.

\bibitem[Kerb]{KernerElectronPairs}
J.~Kerner, \emph{{On bound electron pairs in a quantum wire}}, preprint,
  \url{https://arxiv.org/abs/1708.03753}, 2017.

\bibitem[Kerc]{KernerInteractingPairs}
\bysame, \emph{{On pairs of interacting electrons in a quantum wire}},
  preprint, \url{https://arxiv.org/abs/1801.00696}, 2018.

\bibitem[Kerd]{KernerSurfaceDefects}
\bysame, \emph{{On surface defects and their impact on the superconducting
  phase in quantum wires}}, preprint, \url{https://arxiv.org/abs/1712.07650},
  2017.

\bibitem[Kir08]{KirschInvitation}
W.~Kirsch, \emph{An invitation to random {S}chr\"odinger operators}, Random
  {S}chr\"odinger operators, Panor. Synth\`eses, vol.~25, Soc. Math. France,
  Paris, 2008, pp.~1--119.

\bibitem[KM16]{KM16}
J.~Kerner and T.~M\"uhlenbruch, \emph{Two interacting particles on the
  half-line}, J. Math. Phys. \textbf{57} (2016), no.~2, 023509, 10.
  \MR{3455670}

\bibitem[KM17]{KMBound}
\bysame, \emph{On a two-particle bound system on the half-line}, Rep. Math.
  Phys. \textbf{80} (2017), no.~2, 143--151.

\bibitem[KS99a]{KosSch99}
V.~Kostrykin and R.~Schrader, \emph{Kirchhoff's rule for quantum wires}, J.
  Phys. A: Math. Gen. \textbf{32} (1999), 595--630.

\bibitem[KS99b]{KotSmi99}
T.~Kottos and U.~Smilansky, \emph{Periodic orbit theory and spectral statistics
  for quantum graphs}, Ann. Phys. (NY) \textbf{274} (1999), 76--124.

\bibitem[KS06]{KS06}
V.~Kostrykin and R.~Schrader, \emph{Laplacians on metric graphs: Eigenvalues,
  resolvents and semigroups}, Contemporary Mathematics \textbf{415} (2006),
  201--225.

\bibitem[Kuc04]{Kuc04}
P.~Kuchment, \emph{Quantum graphs. {I}. {S}ome basic structures}, Waves Random
  Media \textbf{14} (2004), S107--S128.

\bibitem[LL63]{LL63}
E.~Lieb and W.~Liniger, \emph{Exact analysis of an interacting {B}ose gas. {I}.
  {T}he general solution and the ground state}, Physical Review \textbf{130}
  (1963), 1605--1616.

\bibitem[LM77]{LeiMyr77}
J.M. Leinaas and J.~Myrheim, \emph{On the theory of identical particles}, Il
  Nouvo Cimento B \textbf{37} (1977), no.~1, 1--23.

\bibitem[LS02]{LiebSeiringerProof}
E.~H. Lieb and R.~Seiringer, \emph{Proof of {B}ose-{E}instein condensation for
  dilute trapped gases}, Phys. Rev. Lett. \textbf{88} (2002), 170409.

\bibitem[LVZ03]{LVZ03}
J.~Lauwers, A.~Verbeure, and V.A. Zagrebnov, \emph{Proof of {B}ose-{E}instein
  condensation for interacting gases with a one-particle gap}, Journal of
  Physics A \textbf{36} (2003), 169--174.

\bibitem[LW79]{LanWil79}
L.~J. Landau and I.~F. Wilde, \emph{On the {B}ose-{E}instein condensation of an
  ideal gas}, Comm. Math. Phys. \textbf{70} (1979), no.~1, 43--51.

\bibitem[McG64]{McG64}
J.~B. McGuire, \emph{Study of exactly soluble one-dimensional {$N$}-body
  problems}, J. Mathematical Phys. \textbf{5} (1964), 622--636.

\bibitem[MP95]{MP95}
Yu.~B. Melnikov and B.~S. Pavlov, \emph{Two-body scattering on a graph and
  application to simple nanoelectronic devices}, J. Math. Phys. \textbf{36}
  (1995), 2813--2825.

\bibitem[MR04a]{MR04}
P.~A. Martin and F.~Rothen, \emph{Many-{B}ody {P}roblems and {Q}uantum {F}ield
  {T}heory}, Springer {V}erlag, Berlin-Heidelberg, 2004.

\bibitem[MR04b]{MinRag04}
M.~Mintchev and E.~Ragoucy, \emph{Interplay between {Z}amolodchikov-{F}addeev
  and reflection-transmission algebras}, J. Phys. A \textbf{37} (2004), no.~2,
  425--431, Special issue on recent advances in the theory of quantum
  integrable systems.

\bibitem[Noj14]{Noj14}
D.~Noja, \emph{Nonlinear {S}chr\"odinger equation on graphs: recent results and
  open problems}, Philos. Trans. R. Soc. Lond. Ser. A Math. Phys. Eng. Sci.
  \textbf{372} (2014), no.~2007, 20130002, 20.

\bibitem[PF92]{PasturFigotin}
L.~Pastur and A.~Figotin, \emph{Spectra of random and almost-periodic
  operators}, Grundlehren der Mathematischen Wissenschaften, vol. 297,
  Springer-Verlag, Berlin, 1992.

\bibitem[PO56]{PO56}
O.~Penrose and L.~Onsager, \emph{{B}ose-{E}instein condensation and liquid
  helium}, Physical Review \textbf{104} (1956), 576--584.

\bibitem[QU16]{QUnruh}
F.~Queisser and W.~G. Unruh, \emph{Long-lived resonances at mirrors}, Phys.
  Rev. D \textbf{94} (2016), 116018, 9.

\bibitem[Rot83]{Rot83}
J.-P. Roth, \emph{Spectre du laplacien sur un graphe}, C. R. Acad. Sci. Paris
  S\'er. I Math. \textbf{296} (1983), 793--795.

\bibitem[Sab14]{Sab14}
M.~Sabri, \emph{Anderson localization for a multi-particle quantum graph}, Rev.
  Math. Phys. (2014), no.~1, 1350020, 52.

\bibitem[Sch06]{SchwablSM}
F.~Schwabl, \emph{Statistical mechanics}, Springer-Verlag, 2006.

\bibitem[Sch09]{Sch09}
R.~Schrader, \emph{Finite propagation speed and causal free quantum fields on
  networks}, J. Phys. A \textbf{42} (2009), no.~49, 495401, 39.

\bibitem[Sto95]{StolzPoisson}
G.~Stolz, \emph{Localization for random {S}chr\"odinger operators with
  {P}oisson potential}, Ann. Inst. H. Poincar\'e Phys. Th\'eor. \textbf{63}
  (1995), no.~3, 297--314.

\bibitem[Sto01]{stollmann2001caught}
P.~Stollmann, \emph{Caught by disorder: bound states in random media}, vol.~20,
  Springer Science \& Business Media, 2001.

\bibitem[Ver11]{VerbeureBook}
A.F. Verbeure, \emph{Many-body boson systems}, Theoretical and Mathematical
  Physics, Springer-Verlag London, Ltd., London, 2011, Half a century later.

\end{thebibliography}

\end{document}